\documentclass[12pt,onecolumn,draftclsnofoot,journal]{IEEEtran}

\usepackage{amsmath,dsfont,bbm,epsfig,amssymb,amsfonts,amstext,verbatim,amsopn,cite,subfigure}
\usepackage{multirow,multicol,lipsum,xfrac}
\usepackage{amsthm}
\usepackage{mathtools,amsthm}
\usepackage{url}
\usepackage{amsfonts}
\usepackage{epsfig}
\usepackage[font=small]{caption}
\usepackage{psfrag}	
\usepackage{algorithm,algorithmic}
\usepackage{pifont}
\usepackage[utf8]{inputenc}
\usepackage[T1]{fontenc}  
\usepackage[nolist]{acronym}
\usepackage{paralist}
\usepackage{enumitem}
\usepackage{bbm}
\usepackage[process=auto]{pstool}
\usepackage{float}

\usepackage{tikz,pgfplots}
\usetikzlibrary{shapes.misc,arrows,positioning}

\usepackage[utf8]{inputenc} 
\usepackage{amsmath,dsfont,bbm,epsfig,amssymb,amsfonts,amstext,verbatim,amsopn,cite,subfigure}
\usepackage{multirow,multicol,lipsum,xfrac}
\usepackage{amsthm}
\usepackage{mathtools,amsthm}
\usepackage{url}
\usepackage{amsfonts}
\usepackage{tikz,pgfplots}
\usepackage[nolist]{acronym}

\newcommand{\setR}{\mathbbmss{R}}

\newcommand{\setS}{\mathbbmss{S}}

\newcommand{\setC}{\mathbbmss{C}}
\newcommand{\setU}{\mathbbmss{U}}
\newcommand{\setD}{\mathbbmss{D}}

\newcommand{\Ex}[2]{ \mathbbm{E}_{#2} \left\lbrace #1 \right\rbrace }

\newcommand{\rmg}{\mathrm{g}}

\newcommand{\her}{\mathsf{H}}

\newcommand{\argmin}{\mathop{\mathrm{argmin}}}
\newcommand{\argmax}{\mathop{\mathrm{argmax}}} 
\newcommand{\subto}{\mathop{\mathrm{subject \; to}\;}}

\newcommand{\maP}{\mathcal{P}}
\newcommand{\mae}{\mathcal{E}}

\newcommand{\man}{\mathcal{N}}

\newcommand{\mao}{\mathcal{O}}

\newcommand{\maA}{\mathcal{A}}
\newcommand{\mac}{\mathcal{C}}

\newcommand{\bxx}{\mathbf{x}}
\newcommand{\mh}{\mathbf{h}}
\newcommand{\bss}{\mathbf{s}}

\newcommand{\buu}{\mathbf{u}}
\newcommand{\bmm}{\mathbf{c}}
\newcommand{\bvv}{\mathbf{v}}
\newcommand{\byy}{\mathbf{y}}
\newcommand{\bgg}{\mathbf{g}}

\newcommand{\rms}{\mathrm{s}}

\newcommand{\bh}{{\mathbf{h}}}
\newcommand{\btheta}{{\boldsymbol{\theta}}}

\newcommand{\bxi}{{\boldsymbol{\xi}}}
\newcommand{\bmu}{\boldsymbol{\mu}}
\newcommand{\bphi}{\boldsymbol{\phi}}
\newcommand{\bpsi}{\boldsymbol{\psi}}
\newcommand{\bx}{{\boldsymbol{x}}}

\newcommand{\xx}{\mathrm{x}}

\newcommand{\set}[1]{\left\lbrace#1\right\rbrace}
\newcommand{\Diag}[1]{\mathrm{Diag}\left\lbrace #1 \right\rbrace}
\newcommand{\brc}[1]{\left( #1 \right) }

\newcommand{\dbc}[1]{\left[ #1 \right] }

\newcommand{\br}{{\mathbf{c}}}
\newcommand{\bt}{{\mathbf{t}}}

\newcommand{\baa}{{\mathbf{a}}}

\newcommand{\trp}{\mathsf{T}}

\newcommand{\mC}{\mathbf{C}}

\newcommand{\mR}{\mathbf{R}}

\newcommand{\mI}{\mathbf{I}}

\newcommand{\mG}{\mathbf{G}}
\newcommand{\mQ}{\mathbf{Q}}

\newcommand{\mU}{\mathbf{U}}

\newcommand{\mTheta}{\mathbf{\Theta}}
\newcommand{\mXi}{\mathbf{\Xi}}
\newcommand{\mY}{\mathbf{Y}}

\newcommand{\mV}{\mathbf{V}}
\newcommand{\mW}{{\mathbf{W}}}
\newcommand{\mSigma}{\boldsymbol{\Sigma}}

\newcommand{\mT}{\mathbf{T}}
\newcommand{\mH}{\mathbf{H}}

\newcommand{\itr}[1]{^{\left( #1 \right)} }

\newcommand{\norm}[1]{\lVert #1 \rVert}

\newcommand{\abs}[1]{\lvert #1 \rvert}
\newcommand{\tr}[1]{\mathrm{tr} \{ #1 \}}

\newtheoremstyle{mystyle}
{}
{}
{\it}
{}
{\bfseries}
{:}
{ }
{}

\theoremstyle{mystyle}

\newtheorem{proposition}{Proposition}
\newtheorem{remark}{Remark}

\newtheorem{example}{Example}[section]

\newcommand\algorithmicinput{\textbf{Input:}}
\newcommand\INPUT{\item[\algorithmicinput]}
\newcommand\algorithmicoutput{\textbf{Output:}}
\newcommand\OUTPUT{\item[\algorithmicoutput]}

\newcounter{bar}

\begin{document}

\begin{acronym}
	\acro{mimo}[MIMO]{multiple-input multiple-output}
	\acro{csi}[CSI]{channel state information}
	\acro{awgn}[AWGN]{additive white Gaussian noise}
	\acro{iid}[i.i.d.]{independent and identically distributed}
	\acro{uts}[UTs]{user terminals}
	\acro{ps}[PS]{parameter server}
	\acro{irs}[IRS]{intelligent reflecting surface}
	\acro{tas}[TAS]{transmit antenna selection}
	\acro{glse}[GLSE]{generalized least square error}
	\acro{rhs}[r.h.s.]{right hand side}
	\acro{lhs}[l.h.s.]{left hand side}
	\acro{wrt}[w.r.t.]{with respect to}
	\acro{rs}[RS]{replica symmetry}
	\acro{mac}[MAC]{multiple access channel}
	\acro{np}[NP]{non-deterministic polynomial-time}
	\acro{papr}[PAPR]{peak-to-average power ratio}
	\acro{rzf}[RZF]{regularized zero forcing}
	\acro{snr}[SNR]{signal-to-noise ratio}
	\acro{sinr}[SINR]{signal-to-interference-and-noise ratio}
	\acro{svd}[SVD]{singular value decomposition}
	\acro{mf}[MF]{matched filtering}
	\acro{gamp}[GAMP]{generalized AMP}
	\acro{amp}[AMP]{approximate message passing}
	\acro{vamp}[VAMP]{vector AMP}
	\acro{map}[MAP]{maximum-a-posterior}
	\acro{ml}[ML]{maximum likelihood}
	\acro{mse}[MSE]{mean squared error}
	\acro{mmse}[MMSE]{minimum mean squared error}
	\acro{ap}[AP]{average power}
	\acro{ldgm}[LDGM]{low density generator matrix}
	\acro{tdd}[TDD]{time division duplexing}
	\acro{rss}[RSS]{residual sum of squares}
	\acro{rls}[RLS]{regularized least-squares}
	\acro{ls}[LS]{least-squares}
	\acro{erp}[ERP]{encryption redundancy parameter}
	\acro{zf}[ZF]{zero forcing}
	\acro{ta}[TA]{transmit-array}
	\acro{ofdm}[OFDM]{orthogonal frequency division multiplexing}
	\acro{dc}[DC]{difference of convex}
	\acro{bcd}[BCD]{block coordinate descent}
	\acro{mm}[MM]{majorization-maximization}
	\acro{bs}[BS]{base-station}
	\acro{aircomp}[AirComp]{over-the-air-computation}
	\acro{ULA}[ULA]{uniform linear array}
	\acro{FL}[FL]{federated learning}
	\acro{otaFL}[OTA-FL]{over-the-air federated learning}
	\acro{los}[LoS]{line-of-sight}
	\acro{nlos}[NLoS]{non-line-of-sight}
	\acro{AoA}[AoA]{angle of arrival}
	\acro{CNN}[CNN]{convolutional neural network}
	\acro{sgd}[SGD]{stochastic gradient descent}
	\acro{aircomp}[AirComp]{over-the-air computation}
\end{acronym}

\title{Matching Pursuit Based Scheduling for Over-the-Air Federated Learning}

\author{
\IEEEauthorblockN{
Ali Bereyhi, \textit{Member IEEE}, 
Adela Vagollari, \textit{Student Member IEEE}, 
Saba Asaad, \textit{Member IEEE}, 
Ralf R. M\"uller, \textit{Senior Member IEEE}, 
Wolfgang Gerstacker, \textit{Senior Member IEEE}, and
H. Vincent Poor, \textit{Life Fellow IEEE}
\thanks{Ali Bereyhi, Adela Vagollari, Saba Asaad, Ralf R. M\"uller and Wolfgang Gerstacker are with the Institute for Digital Communications (IDC) at Friedrich-Alexander-Universität Erlangen-Nürnberg (FAU); email: \texttt{\{ali.bereyhi, adela.vagollari, saba.asaad, ralf.r.mueller, wolfgang.gerstacker\}@fau.de}. H. Vincent Poor is with the Department of Electrical and Computer Engineering at the Princeton University; email: \texttt{poor@princeton.edu}.}
}
}


\IEEEoverridecommandlockouts

\maketitle

\begin{abstract}
This paper develops a class of low-complexity device scheduling algorithms  for over-the-air federated learning via the method of matching pursuit. The proposed scheme tracks closely the \textit{close-to-optimal} performance achieved by difference-of-convex programming, and outperforms significantly the well-known benchmark algorithms based on convex relaxation. Compared to the state-of-the-art, the proposed scheme poses a drastically lower computational load on the system: For $K$ devices and $N$ antennas at the parameter server, the benchmark complexity scales with $ {\brc{N^2+K}^3 + N^6}$ while the complexity of the proposed scheme scales with ${K^p N^q}$ for some $0 < p,q \leq 2$. The efficiency of the proposed scheme is confirmed via numerical experiments on the CIFAR-10 dataset.
\end{abstract}


\IEEEpeerreviewmaketitle

\section{Introduction}
In the light of dramatically increasing numbers of mobile devices and data traffic in the Internet-of-Things era, the need for a paradigm-shift in wireless networks from traditional centralized cloud computing architectures to distributed ones is growing \cite{lyu2020threats,aledhari2020federated,liang2020think,yang2021federated,nguyen2021federated}. By performing data processing at the edge of networks, several shortcomings of cloud computing, such as long latency and network congestion, can be effectively addressed \cite{ananthanarayanan2017real, hung2018videoedge,zhou2019edge}. Notably, edge computing is an appealing technology to perform real-time tasks and make real-time decisions by exploiting the abundant computational resources of the edge servers \cite{shi2016edge,shi2016promise,abbas2017mobile}. Nevertheless, the bandwidth limitations and resource constraints of the wireless channels can pose significant challenges to realizing fast learning \cite{jordan2018communication, chen2018lag,elgabli2020gadmm}. 
One way of overcoming these challenges is to integrate the edge-intelligent network within wireless networks and leverage the superposition property of wireless multiple-access channels \cite{liu2020over}. 

Recently, a new paradigm of distributed machine learning, referred to as \textit{\ac{FL}} has been introduced, in which distributed devices jointly train a shared global machine learning model without sharing their raw data explicitly \cite{li2020federated,konevcny2016federated,bonawitz2019towards}. In essence, \ac{FL} is a collaborative machine learning framework that enables distributed model training from decentralized data under coordination of a \ac{ps} \cite{konevcny2016federated}.  In principle \ac{FL} is performed over a decentralized network as follows: \begin{enumerate}
	\item A \ac{ps} first shares a global model with participating devices in the network. 
	\item Each device performs local model training using its own local dataset to determine the model parameters. It then transmits its trained model parameters to the \ac{ps} while keeping its private data locally within its own device. 
	\item Once the \ac{ps} aggregates the locally trained models, it updates the global model parameters using the aggregated models and broadcasts updated parameters to all edge devices.
\end{enumerate} 
These steps are alternated until the global model parameters converge \cite{bonawitz2019towards,li2020federated,yang2019federated}. Further illustrations can be found through the comprehensive example of \ac{FL} given in Appendix~\ref{App:1}.

Compared to the extreme cases of \textit{centralized} and \textit{individual} learning, \ac{FL} provides a tractable approach to handle a joint learning task over a distributed network. Nevertheless, this tractability comes with some costs which can be roughly categorized into three major forms:
\begin{enumerate}
	\item The \textit{statistical inference} problem in \ac{FL} is more challenging. This follows from the fact that the local datasets in the decentralized setting are not \ac{iid}. Some related discussions in this respect can be found in \cite{sattler2019robust,shoham2019overcoming,lee2020bayesian,luo2021no,gafni2022federated} and the references therein.
	\item Although local datasets are not shared in \ac{FL}, this approach is still vulnerable against \textit{adversarial} attacks in the network. In fact, due to the multiple rounds of model updates over the network, malicious terminals have access to a large list of local models from which they can learn the local datasets via model inversion attacks; see \cite{fredrikson2015model} for an instance of such attacks, and also discussions in \cite{wei2020secFed,kim2021federated} and references therein regarding the privacy challenges in \ac{FL}.
	\item To achieve convergence of the global model, an \ac{FL} algorithm needs to iterate for a rather large number of rounds. This leads to a \textit{high communication cost} in the network which can result in excessive  overhead to the system or even congest it; see for instance \cite{shi2020communication} and the references therein.
\end{enumerate}

These challenges are often addressed separately in the literature. This follows from the fact that they result from different sources of imperfection. In this work, we mainly focus on the third class of challenges, i.e., communication-related challenges, as this class is the key hindrance for \ac{FL} in wireless networks.

\subsection{Over-the-Air Federated Learning}
As mentioned, the key practical challenges in \ac{FL} over wireless networks, often called \textit{\ac{otaFL}}, are the communication-related ones. This comes from the facts that wireless networks are significantly restricted in resources and the communication links are severely impacted by fading. As a result, direct implementation of \ac{FL} in these networks can lead to large communication overhead in the system and high latency; see \cite{yang2020scheduling} and the references therein for further discussion of these issues. Consequently, developing communication-efficient strategies to mitigate communication cost and guarantee satisfactory learning performance is of paramount importance \cite{chen2021communication}. The proposed solutions in the literature often lie in one of the following three approaches: 
\begin{enumerate}
	\item One approach is to minimize the total number of communication rounds between the edge devices and \ac{ps}. With this approach, the main design task is to achieve satisfactory learning accuracy with a minimum number of communication rounds \cite{crane2019dingo, gao2020can,nguyen2020fast}.
	\item Another approach is to reduce the communication overhead per round by exploiting lossy compression techniques such as quantization and pruning \cite{lin2018deep, yuan2018variance,shlezinger2020uveqfed}.
	\item The third approach suggests to accelerate the model aggregation process by leveraging the principles of \textit{over-the-air computation} \cite{zhu2019broadband, yang2020federated, amiri2020machine, xu2021learning}.
\end{enumerate}

The first two approaches treat communication and computation separately via the \textit{transmission-then-aggregation} policy. The latter approach however invokes the idea of analog function computation in sensor networks \cite{nazer2007computation} to address both tasks simultaneously; see Appendix~\ref{App:1} for an illustrative example.


\subsection{Device Scheduling for Over-the-Air Federated Learning}
From the inference point of view, we are interested in increasing the number of devices that share their model parameters. This is in particular of a great interest, as the accuracy of the learned model strongly depends on the size of the aggregated dataset on which the \ac{FL} algorithm is running. The more devices that share their local model parameters, the more reliable the learned model becomes. Although this statement is in general valid for any network, it is not always straightforward in \ac{otaFL} to improve the reliability of the learned model by increasing the number of participating devices. This comes from the wireless channel effects, i.e., fading. The optimal solution to this issue is to design the model transmission and the \ac{FL} algorithm jointly. By this approach, the share of each local model can be designed such that it contributes in the global model learning constructively. This joint approach is however not in general tractable, as the statistical model of the local model parameters in \ac{FL} is not trivial. An alternative sub-optimal solution to this issue is given by \textit{device scheduling} in which some of the devices with potentially less contribution in the global model learning are set off. In this case, performing \ac{FL} and model transmission separately leads to minimal performance degradation compared to the case with all devices being active\footnote{It is worth mentioning that scheduling is in general \textit{sub-optimal}. To illustrate this point, one can think of a real-valued  weight for each device representing its share in the global model learning. The joint design of the \ac{FL} algorithm and transmission scheme determines the optimal choice of these real-valued coefficients while device scheduling restricts them to be either zero or one.}.

Device scheduling has been widely accepted as an efficient \textit{sub-optimal} solution in \ac{otaFL} which avoids performance degradation due to imperfect model aggregation\footnote{Another root of interest on device scheduling in \ac{otaFL} comes from network limitations, such as the total communication overhead, latency and the required bandwidth. We however do not discuss this aspect of device scheduling in this paper.}. The efficiency of this approach has been shown analytically and experimentally through a large set of studies in the literature; see for instance the studies in \cite{yang2020scheduling,shi2021joint,xia2021federated,amiri2020update}. In general, a scheduling policy tries to find the largest collection of local datasets in each communication round while keeping the overall aggregation cost, for a predefined aggregation strategy\footnote{This strategy is designed independent of the channel information. This separation is in fact the root of sub-optimality.}, below a tolerable threshold. In its generic form, device scheduling reduces to an integer programming problem and hence is a \ac{np}-hard problem. Consequently, various sub-optimal algorithms for scheduling are proposed in the literature \cite{yang2020federated} and \cite{chen2020wireless,yang2020energy,Wang2022FLIRS}. These algorithms often either offer good performance at the expense of high computational complexity, e.g., \cite{yang2020federated}, or perform poorly while gaining in terms of complexity, e.g., random scheduling \cite{yang2020scheduling}. This work fills the gap between these two types of scheduling policies by introducing a class of scheduling algorithms with a fair complexity-performance trade-off based on the method of matching pursuit.

\subsection{Contributions}
The device scheduling task is mathematically formulated in the form of a constrained optimization problem: In this problem, the number of active devices is maximized, subject to an inequality constraint on the aggregation cost; see for instance the formulations in \cite{yang2020federated} and \cite{Wang2022FLIRS}. The explicit form of the problem is given by defining a metric for the aggregation cost in the network which depends on the chosen approach for \ac{otaFL}. The key feature of the aggregation cost metric is that it grows with the number of devices and describes the dominant source of error arising in the model aggregation step.

A device scheduling mechanism proposed for a particular aggregation cost metric often extends straightforwardly to other metrics. This follows from the fact that the target optimization problem is of the same \ac{np}-hard form with only the cost metric being different. We hence focus in this work on a recent \ac{otaFL} approach; namely, the approach based on \textit{over-the-air computation}\footnote{Nevertheless, as mentioned through the introductory part, the proposed scheduling scheme is generic and can be extended straightforwardly to other approaches for \ac{otaFL}.} initially proposed in \cite{yang2020federated}. As the main contribution, we develop a class of low-complexity scheduling algorithms based on the method of matching pursuit \cite{mallat1993matching} which incurs a significantly lower computational complexity compared to a benchmark algorithm based on the \ac{dc} programming \cite{yang2020federated}, while degrading the scheduling performance only slightly. In particular, the contributions of this paper can be briefly described as follows:
\begin{itemize}
	\item We formulate the device scheduling task as a sparse support selection problem whose constraint is given by a weighted combination of individual cost constraints. For this problem, we develop a class of greedy algorithms based on the method of matching pursuit. We show that with $K$ devices in the network and an array antenna of size $N$ at the \ac{ps}, the computational complexity of the proposed scheme scales with ${K^p N^q}$ for some $p,q \leq 2$. This is a significant complexity reduction compared to the benchmark scheme in \cite{yang2020federated} which scales with $K\brc{N^2+K}^3 + KN^6$.
	\item The proposed scheduling scheme is parametrized by a set of device weights which need to be tuned. To this end, we invoke the intuitive connection between the original scheduling task and the weighted sparse support selection problem and propose a \textit{subset-cutting} strategy for weighting. Our investigations show that the proposed strategy is robust against variation of design parameters and performs very closely to optimized approaches for weighting. 
	\item We numerically investigate the proposed scheme through a comprehensive set of simulations over the CIFAR-10 dataset and compare the performance with the benchmark. Our investigations demonstrate that the learning performance in this case closely track the close-to-optimal performance achieved by \ac{dc} programming. This finding along with the significantly lower computational complexity implies the efficiency of our proposed scheme for \ac{otaFL}.
\end{itemize}
 
\subsection{Related Work}
The efficiency of device scheduling for \ac{otaFL} is discussed in \cite{yang2020scheduling}. In this study, a standard setting for \ac{otaFL}, i.e., without invoking analog function computation, is considered and the learning performance of three basic schemes for device scheduling are discussed; namely, the random scheduling, the round robin scheme and the so-called \textit{proportional fair} strategy. The investigations in this work reveal a key finding: with links that require high \ac{sinr} thresholds to correctly decode received packets, an optimized scheduling algorithm, e.g., the \textit{proportional fair} strategy, outperforms other schemes while at low \ac{sinr} thresholds, random selection of active devices performs efficiently. There exists further a very low required \ac{sinr} threshold under which device scheduling does not enhance the \ac{FL} convergence rate. This observation agrees with intuition: At very low \ac{sinr} thresholds, virtually all packets are received correctly at the model aggregator, and thus scheduling mechanism is not important. This further agrees with our earlier discussions on the sub-optimality of device scheduling: In the most extreme case of noise-less transmission, the model aggregation strategy and the transmission scheme are designed separately without any performance degradation. As a result, ignoring some local models can only degrade the \ac{FL} performance.

Following the discussions in \cite{yang2020scheduling}, several studies have developed scheduling schemes considering various approaches for \ac{otaFL}; see for instance \cite{amiri2020update,amiri2020machine,yang2020energy,yang2020federated,Wang2022FLIRS,liu2021reconfigurable,xu2021learning} and \cite{wadu2020federated} . The most relevant line of work to our study in this paper is the one given in \cite{yang2020federated}. In this work, a novel aggregation approach for \ac{otaFL} based on the idea of analog function computation, often called \textit{\ac{aircomp}}, is proposed; see \cite{nazer2007computation} and \cite{goldenbaum2013robust,goldenbaum2013harnessing,chen2018uniform,zhu2018mimo,zhai2021hybrid}. Unlike the earlier approaches, in this approach the \ac{ps} utilizes the linear superposition in the uplink multiple access channel\footnote{It is worth mentioning that \ac{aircomp} is only used to \textit{implement} the model aggregation over the air. The aggregation strategy, i.e., the weights for local models, is still designed individually.} and determines the updated global model in each communication round directly from its received signal via a linear receiver.  

Since its proposal, \ac{otaFL} via \ac{aircomp} has been investigated in various network settings. The studies in \cite{Wang2022FLIRS} and \cite{liu2021reconfigurable} extend the idea to  \ac{irs}-aided \ac{mimo} networks concluding that by leveraging large \acp{irs}, data aggregation in \ac{aircomp}-based \ac{otaFL} can be significantly fastened. The study in \cite{cao2021optimized} further considers this approach and develops a novel power-control policy to enhance the learning performance. In \cite{sery2020analog}, the authors propose an algorithm for \ac{otaFL} via \ac{aircomp}, considering a risk-minimization task as the target learning problem. The developed algorithm is shown to approach the convergence rate of the centralized gradient-descent, when the network dimensions grow large.

Similar to other approaches for \ac{otaFL}, \ac{otaFL} via \ac{aircomp} requires a joint design of aggregation strategy and the transmission scheme, in order to perform efficiently. As a result, device scheduling\footnote{As an analytically-tractable, but generally sub-optimal, solution.} is often used to achieve a satisfactory learning performance. For device scheduling in \ac{otaFL} via \ac{aircomp}, the aggregation cost can be determined via the \ac{mse} between the model aggregated by the linear receiver at the \ac{ps} and the target aggregation specified by the predefined aggregation strategy. Considering this cost metric, the authors in \cite{yang2020federated} develop a mechanism for joint scheduling and receiver design based on \ac{dc} programming. The proposed mechanism in \cite{yang2020federated} has been further extended to other network settings; see for instance \cite{Wang2022FLIRS} and \cite{liu2021reconfigurable} for device scheduling in \ac{irs}-aided \ac{mimo} networks.

\subsection{Notation and Organization}
Scalars, vectors and matrices are represented with non-bold, bold lower-case, and bold upper-case letters, respectively. The transposed and the transposed conjugate of $\mH$ are denoted by $\mH^{\trp}$ and $\mH^{\her}$, respectively,  and $\mI_N$ is an $N\times N$ identity matrix. Trace and rank of the matrix $\mH$ are shown by $\tr{\mH}$ and $\mathrm{rank}\brc{\mH}$, respectively. The $\ell_p$-norm of $\bx$ is denoted by $\norm{\bx}_p$. For the particular case of the $\ell_2$-norm, the subscript is dropped, i.e., $\norm{\bx} = \norm{\bx}_2$. The sets $\setR$ and $\setC$ refer to the real axis and the complex plane, respectively. The notation $\mathcal{CN}\brc{\eta,\sigma^2}$ represents the complex Gaussian distribution with mean $\eta$ and variance $\sigma^2$. For the sake of brevity, $\set{1,\ldots,N}$ is shortened to $\dbc{N}$.

The rest of this manuscript is organized as follows: The problem under study is formulated in Section~\ref{sec:system-model}. Section~\ref{sec:Schedul} presents the proposed scheduling scheme along with discussions of its performance and complexity. The detailed derivation of the scheme is then given in Section~\ref{sec:derive}. Section~\ref{sec:NumRes} provides several numerical experiments. Finally, the manuscript is concluded in Section~\ref{sec:conc}.

\section{Problem Formulation}
\label{sec:system-model}
Consider a decentralized setting with $K$ single-antenna edge devices and a single \ac{ps}. The \ac{ps} is equipped with an antenna array of size $N$. The dataset $\setD$ is distributed among the edge devices. This means that device $k$ for $k\in \dbc{K}$ has access to a local dataset $\setD_k$ from which it determines its local vector of model parameters $\btheta_k\in\setR^D$. The model parameters are shared with the \ac{ps} over the uplink multiple access channel. The \ac{ps} aims to determine an update for the global model based on the shared local models.

We assume that the network operates in the \ac{tdd} mode. As a result, the uplink and downlink channels are assumed to be reciprocal. Prior to sharing the local datasets, the devices send pilot sequences over the uplink channels. The \ac{ps} utilizes its received signal to acquire the \ac{csi} of the devices. As a result, the \ac{csi} of all devices is available at the \ac{ps}, and each edge device has access to its own \ac{csi}. For the sake of compactness, we assume that the \ac{csi} acquisition is carried out perfectly, i.e., we assume that the pilot sequences are mutually orthogonal and that the channel estimation error is negligible. The coherence time interval, in which the \ac{csi} remains unchanged, is further assumed to include multiple symbol time intervals. Using the acquired \ac{csi}, the \ac{ps} schedules a subset of devices to be active in model sharing. Our main goal in this work is to design the scheduling protocol. 

\subsection{System Model}
We now focus on a particular communication round and assume that the \ac{ps} performs device scheduling. Let $\setS\subseteq \dbc{K}$ denote the subset of devices that are scheduled to share their model in the current communication round. Considering the \textit{federated averaging} algorithm for \ac{FL}\footnote{This is the most classical approach for \ac{FL}; see also Appendix~\ref{App:1}.} \cite{mcmahan2017communication}, the ultimate goal of the \ac{ps} is to determine an updated global model $\btheta$ as
\begin{align}
	{\btheta}  \brc{\setS} = \sum_{k\in\setS} \phi_k \btheta_k, \label{eq:theta}
\end{align}
for some predefined positive weights $\phi_k$ which are proportional to the size of the local datasets.

The selected device $k$ uploads its local model parameters in $D$ symbol intervals after applying a linear operation on them. More precisely, it transmits in the $d$-th symbol interval
\begin{align}
	\xx_k \brc{d} = \psi_k \theta_{k,d}, 
\end{align}
where $\theta_{k,d}$ denotes the $d$-th entry of $\btheta_k$, and $\psi_k$ is a scalar that satisfies the per-device transmit power constraint $\abs{\psi_k}^2 \leq P$. 

In the sequel, we drop the time interval index, and focus on a single transmission time interval. In this case, the signal received by the \ac{ps} is given by
\begin{align}
	\byy = \sum_{k \in \setS} \xx_k \bh_k + \bxi
\end{align}
where $\bxi$ is a complex zero-mean \ac{awgn} process with variance $\sigma^2$, i.e., $\bxi\sim\mac\man\brc{\boldsymbol{0},\sigma^2 \mI_N}$, and $\mh_k\in \setC^N$ denotes the channel coefficient vector between device $k$ and the \ac{ps}.

The \ac{ps} ultimately intends to compute the global model parameter according to the aggregation strategy in \eqref{eq:theta}. This means that in a particular symbol time interval, its target function is
\begin{align}
{\theta} \brc{\setS} = \sum_{k \in \setS} \phi_k \theta_{k}.
\end{align}
Due to noise and fading processes in uplink channels, the \textit{aggregation} of the selected local model parameters is noisy, and hence the \ac{ps} can only determine an estimate of ${\theta} \brc{\setS}$. To this end, the \ac{ps} invokes the idea of \textit{analog function computation} and utilizes the linear superposition of the multiple access channel by determining the target function directly from the received signal via a linear receiver. This means that it calculates the estimate
\begin{align}
	\hat{\theta} \brc{\bmm , \setS, \eta , \bpsi} = \frac{\bmm^\her \byy}{\sqrt{\eta}} = \sum_{k\in\setS} \frac{\psi_k \theta_{k}}{\sqrt{\eta}} \bmm^\her \bh_k + \frac{\bmm^\her \bxi}{\sqrt{\eta}} \label{eq:hatTheta}
\end{align}
for some linear receiver $\bmm\in \setC^N$ and the power factor $\eta$. In \eqref{eq:hatTheta}, we define the vector $\bpsi$ as
\begin{align}
	\bpsi = \dbc{\psi_1,\ldots,\psi_K}^\trp.
\end{align}
The estimate in \eqref{eq:hatTheta} incurs some error compared to the target function. This error describes the distortion imposed through \textit{model aggregation} in this \ac{otaFL} approach. 

\subsection{Uplink Coordination via Zero-Forcing}
\label{sec:ZF}
As mentioned, each edge device only knows its own \ac{csi}. The \ac{ps} hence coordinates the devices with respect to their weights in the aggregation strategy, i.e., the $\phi_k$'s, using the zero-forcing strategy: Let the receiver and the subset of active devices be set to $\bmm$ and $\setS$, respectively. The \ac{ps}, at the beginning of the communication round, sets the power factor to $\eta = \eta_{\rm ZF}$, where
\begin{align}
	\eta_{\rm ZF} = P \min_{k\in \setS} \dfrac{\abs{\bh_k^\her \bmm}^2}{\phi_k^2}, \label{eq:fix_eta}
\end{align}
and broadcasts it, along with $\bmm$ in the network over a rate-limited noiseless feedback channel. Upon the reception of $\eta_{\rm ZF} $ and $\bmm$, each selected device sets its transmit weight as
\begin{align}
	\psi_{\mathrm{ZF}, k} = \sqrt{	\eta_{\rm ZF} } \phi_k \dfrac{\mh_k^\her \bmm}{\abs{\bh_k^\her \bmm}^2}.
\end{align}
It is straightforward to show that the above choices of $\eta$ and $\psi_k$ satisfy the power constraints at the devices and also construct the desired superposition over the air in the absence of noise in the uplink transmission. This coordination approach can be observed as \textit{zero-forcing}, since the \ac{ps} determines the transmit weights of the devices and the receiver power factor, such that the uplink channel is inverted.

From the signal processing point of view, it is well-known that with respect to the \ac{mse} between the target function and its estimate, zero-forcing coordination is \textit{sub-optimal}. The optimal approach is to derive the \ac{mmse} coordination strategy that minimizes the \ac{mse} between ${\theta} \brc{\setS}$ and $\hat{\theta} \brc{\bmm , \setS, \eta , \bpsi}$. Analytical derivation of the \ac{mmse} strategy is however intractable, due to the nontrivial statistical model of the local model parameters. This point is discussed in greater detail in Appendix~\ref{app:MMSE}.

\begin{remark}
	In \cite{yang2020federated}, the zero-forcing approach for coordination is considered as the optimal strategy, while the \ac{mse} between the target function and its estimate is considered as the design metric. Although this consideration does not impact the correctness of the main results, we believe it is due to a misconception. We address this point in Appendix~\ref{app:MMSE}.
\end{remark}

\subsection{Aggregation Cost Metric}
As mentioned in the introduction, the design of a device scheduling scheme requires defining a metric for the aggregation cost in the network. To this end, we follow the proposed strategy in the literature, e.g., \cite{yang2020federated,Wang2022FLIRS}, and model the aggregation cost via the computation error. Considering zero-forcing coordination at the \ac{ps}, we define the \textit{computation error} to be the \ac{mse} between the estimated computation and the target function, i.e., 
\begin{align}
	\mae \brc{\bmm,\setS} = \Ex{\abs{\hat{\theta} \brc{\bmm , \setS, \eta_{\rm ZF} , \bpsi_{\rm ZF}} -{\theta}\brc{\setS} }^2 }{ }.
\end{align}
By replacing $\eta_{\rm ZF}$ from \eqref{eq:fix_eta} into the definition of the computation error, it is straightforwardly shown that 
\begin{align}
	\mae \brc{\bmm, \setS}=  \frac{\sigma^2}{P}\max_{k\in \setS} \phi_k^2 \dfrac{\norm{\bmm}^2}{\abs{\bh_k^\her \bmm}^2}. \label{eq:E_func}
\end{align}
This error serves as the metric for the aggregation cost in the network in the sequel. 

Proposition~\ref{Prop:1} shows that the computation error has the generic property exhibited by a metric of aggregation cost. This means that it either increases or remains unchanged, as the number of active devices increases. Consequently, we design the scheduling protocol and the linear receiver, such that the above design metric does not exceed a maximum tolerable value.

\begin{proposition}
	\label{Prop:1}
	Consider subsets $\setS_i\subseteq \dbc{K}$ for $i\in\set{1,2}$. Denote the computation error that is minimized with respect to the linear receiver by $\mae_i^{\min}$, i.e.,
	\begin{align}
		\mae_i^{\min} = \min_{\bmm\in \setC^N} \mae \brc{\bmm, \setS_i}.
	\end{align}
	If  $\setS_1 \subseteq \setS_2$; then, we have $\mae_1^{\min} \leq \mae_2^{\min}$.
\end{proposition}
\begin{proof}
	The proof is straightforwardly derived from \eqref{eq:E_func}. By definition, we can write for any linear operator $\bmm$ 
	\begin{subequations}
		\begin{align}
			\mae \brc{\bmm, \setS_2} &=  \frac{\sigma^2}{P}\max_{k\in \setS_2} \phi_k^2 \dfrac{\norm{\bmm}^2}{\abs{\bh_k^\her \bmm}^2}\\
			&= \max \set{ \frac{\sigma^2}{P}\max_{k\in \setS_1} \phi_k^2 \dfrac{\norm{\bmm}^2}{\abs{\bh_k^\her \bmm}^2} , \frac{\sigma^2}{P}\max_{k\in \setS_2-\setS_1} \phi_k^2 \dfrac{\norm{\bmm}^2}{\abs{\bh_k^\her \bmm}^2}}\\
			&= \max \set{ \mae \brc{\bmm, \setS_1} , \mae \brc{\bmm, \setS_2-\setS_1} } \\
			&\geq \mae \brc{\bmm, \setS_1}.
		\end{align}
	\end{subequations}
	Considering the functions $\mae \brc{\bmm, \setS_i}: \setC^N \mapsto \setR $ for $i\in\set{1,2}$, we use the fact that above inequality is valid for all $\bmm\in\setC^N$ and write
	\begin{align}
		\min_{\bmm\in \setC^N} \mae \brc{\bmm, \setS_2}  \geq \min_{\bmm\in \setC^N}\mae \brc{\bmm, \setS_1}
	\end{align}
	which concludes the proof.
\end{proof}

\subsection{Scheduling Problem}
Proposition~\ref{Prop:1} verifies that the defined computation error increases as the number of selected terminals grows in the network. This is in contrast to a reliability metric for \ac{FL}, which improves as the size of the collected dataset grows. In this respect, device scheduling finds a fair trade-off by solving a joint optimization: It tries to maximize the number of selected edge devices\footnote{One may initially propose to maximize the size of the aggregated dataset, i.e., $\sum_{k \in \setS}\setD_k$, instead of the number of selected devices. Nevertheless, one should note that the weights $\phi_k$ are chosen with respect to the size of local datasets. This means that the impact of different dataset sizes is directly considered in the predefined model aggregation strategy, and hence is not further addressed in scheduling.}, subject to the constraint that the computation error is kept below a maximum tolerable level. Mathematically, we can represent this problem as
\begin{align}
	&\max_{\bmm\in\setC^N, \setS} \; \abs{\setS} \\
	&\begin{array}{ll}
		\subto & \displaystyle \max_{\setS} \; \mae \brc{\bmm, \setS} \leq \epsilon\\
		& \setS \subseteq \dbc{K}
	\end{array}\nonumber	.
\end{align}
for some positive real $\epsilon$ representing the maximum tolerable computation error.
By substituting \eqref{eq:E_func} into the optimization, the scheduling problem reduces to
\begin{align}
	&\max_{\bmm\in\setC^N, \setS} \; \abs{\setS} \label{eq:main} \\
	&\begin{array}{lll}
		\subto & \phi_k^2 {\norm{\bmm}^2} - \gamma {\abs{\bh_k^\her \bmm}^2} \leq 0  & \forall k\in\setS \\
		& \setS \subseteq \dbc{K} &
	\end{array}\nonumber	.
\end{align}
where we define $\gamma = P\epsilon/\sigma^2$. 

The optimization problem in \eqref{eq:main} reduces to an \ac{np} hard problem, and hence its solution cannot generally be found within a feasible time. As a result, sub-optimal schemes are developed to approximate the optimal scheduling in polynomial time. An example of such approaches is given below in Example~\ref{ex:Alg}. This example is of particular interest, as it describes the \textit{close-to-optimal} and \textit{benchmark} scheduling policies which are considered as the references in this work. 

\begin{example}
	\label{ex:Alg}
	The classical approach is to convert the scheduling task into a sparse recovery problem with concave constraints. To this end, the authors in \cite{yang2020federated} show that solving the original problem in \eqref{eq:main} is equivalent to solving the following optimization:
	\begin{align}
		&\min_{\mC \in\setC^{N \times N}, \bss \in \setR_+^K} \; \norm{\bss}_0  \\ 
		&\begin{array}{lll}
			\subto & \tr{\mC} - \displaystyle \frac{\displaystyle \gamma}{\displaystyle \phi_k^2} \bh_k^\her \mC\bh_k \leq  \rms_k  & \forall k\in\dbc{K}\\
			& \mC \succeq \boldsymbol{0}_{N \times N} & \\
			&\tr{\mC} \geq 1 & \\
			&\mathrm{rank} \brc{\mC} = 1 &
		\end{array}\nonumber	.
	\end{align}
The solution of this equivalent form is then approximated by different approaches. In this paper, we focus on two techniques; namely, the benchmark and the close-to-optimal approach. 

In the benchmark, the $\ell_0$-norm is replaced by the $\ell_1$-norm and the rank constraint is dropped. The resulting problem is then solved using alternating optimization, and the resulting $\mC$ is approximated by a rank-one matrix using a low-rank approximation technique \cite{tropp2017practical}. 

The close-to-optimal approach replaces the $\ell_0$-norm objective and the rank-one constraint via \ac{dc} expressions. It then uses tools from \ac{dc} programming \cite{tao1997convex,le1997solving} to approximate the solution of the determined alternative form. We call this approach close-to-optimal, as the numerical investigations in \cite{yang2020federated} show that for small dimensions, where exponential search is feasible, the \ac{dc} programming-based approach is consistent with the optimal solution.
\end{example}

In this work, we deviate from the benchmark and close-to-optimal approaches described in Example~\ref{ex:Alg} and propose a low-complexity and highly efficient algorithm that directly applies matching pursuit to address the scheduling task in \eqref{eq:main}.

\section{Scheduling via Matching Pursuit}
\label{sec:Schedul}
The scheduling problem in \eqref{eq:main} can be considered to be a sparse recovery problem with an unconventional constraint. We intend to find the smallest subset of devices, such that by setting them off, the \textit{minimum} computation error\footnote{That is minimized over all linear receivers $\bmm$.} falls below the maximum tolerable error. Defining $\bar{\rms}_k$ to be the \textit{inactivity} of device $k$, i.e., $\bar{\rms}_k = 0$ when $k\in\setS$ and $\bar{\rms}_k=1$ if $k\notin\setS$, one can observe the scheduling task as the problem of finding the sparsest $\bar{\bss} = \dbc{\bar{\rms}_1,\ldots,\bar{\rms}_K}$, such that the computation error constraint is satisfied. Although the scheduling constraint is different from the conventional least-squares constraint in sparse recovery, similar greedy approaches can be developed for this problem. 

The most common greedy approach for sparse recovery is \textit{matching pursuit} which is widely used in the context of compressive sensing; see for instance \cite{foucart2013invitation,donoho2006compressed,candes2006robust,tropp2007signal,rebollo2002optimized,wang2012generalized,cai2011orthogonal,needell2009cosamp}. Intuitively, matching pursuit follows a step-wise greedy approach to solve a regression problem: it starts with the sparsest vector of regression coefficients, i.e., the all-zero vector, and gradually add new indices to the support, such that the residual sum-of-squares\footnote{More precisely, an upper-bound on the residual sum-of-squares.} is maximally suppressed. For linear regression, it is shown that to this end, the selected indices in each iteration should be those that have maximal correlation with the regression error achieved in the previous iteration\footnote{This is where the appellation comes from.}.

For the scheduling problem in \eqref{eq:main}, we can develop the same framework. As the initial choice for $\setS$, we let all the devices be active. We then iteratively remove elements from $\setS$ and update $\bmm$ in each iteration such that the computation error is maximally decreased. The resulting algorithm is presented in Algorithm~\ref{alg:MPS}. We skip the detailed derivation at this stage and leave it for Section~\ref{sec:derive}.  For notational compactness, we denote the relation between the inputs and outputs of this algorithm as
\begin{align}
	\brc{\bmm , \setS} = \maA_1\brc{\mH, \bphi},
\end{align}
where $\bphi$ is the vector of device coefficients, i.e., 
\begin{align}
	\bphi = \dbc{\phi_1,\ldots,\phi_K}^\trp,
\end{align}
and $\mH$ is the uplink channel matrix, i.e., 
\begin{align}
	\mH = \dbc{\bh_1,\ldots,\bh_K}.
\end{align}

\begin{algorithm}[H]
	\caption{Matching Pursuit Based Scheduling $\maA_1\brc{\cdot}$} 
	\label{alg:MPS}
	\begin{algorithmic}
		\INPUT The updating strategy $\Pi\brc{\cdot}$, and the real positive scalars $\phi_k$ for $k\in \dbc{K}$.
		\REQUIRE{ Set $\setS\itr{0} = \dbc{K}$, $i\itr{1} = \emptyset$, $\mW\itr{1} = \mI_K$, $\Delta\itr{1} = +\infty$ and $F_k\itr{1} = 1$ for $k\in\dbc{K}$.}
		\WHILE{ $\Delta\itr{t} > 0$ }{
			\STATE Set $\setS\itr{t} = \setS\itr{t-1} - \set{i\itr{t}}$
			\STATE Update $\mW\itr{t}$ by setting $\dbc{\mW\itr{t}}_{i\itr{t},i\itr{t}} = 0$ and
			\begin{align}
				\dbc{\mW\itr{t}}_{k,k} = \Pi \brc{ F_k\itr{t} }
			\end{align}
			for all $k\in\setS\itr{t}$.
			\STATE Find the SVD
			\begin{align}
				\mH \sqrt{\mW\itr{t}} = \mV\itr{t} \mSigma\itr{t} \mU^{\brc{t} \her}
			\end{align}
			and let $n\itr{t}$ be the index of the largest singular value.
			\STATE Let $\bmm\itr{t} = \bvv\itr{t}_{n\itr{t}}$, i.e., the $n\itr{t}$-th column of $ \mV\itr{t}$.
			\STATE Update the \textit{constraint indicator} $F_k\itr{t}$ for device $k \in \setS\itr{t} $ as
			\begin{align}
				F_k\itr{t}  = \phi_k^2  - \gamma {\abs{\bh_k^\her \bmm\itr{t} }^2}. \label{eq:F_k}
			\end{align}
			\STATE Find the next device index as
			\begin{align}
				i\itr{t+1} = \argmax_{k \in \setS\itr{t} } F_k\itr{t}.
			\end{align}
			\STATE Update the maximal constraint as
			\begin{align}
				\Delta\itr{t+1} =  \max_{k \in \setS\itr{t} } F_k\itr{t}.
			\end{align}
			\STATE Let $t \leftarrow t+1$.
		}
		\ENDWHILE
		\OUTPUT $\bmm\itr{T}$ and $\setS\itr{T}$ with $T$ being the last iteration.
	\end{algorithmic} 
\end{algorithm}

The key components of Algorithm~\ref{alg:MPS} are the weighting matrix $\mW\itr{t}\in \setR^{K\times K}$ and the updating strategy $\Pi\brc{\cdot}:\setR \mapsto \setR$. The weighting matrix $\mW\itr{t}\in \setR^{K\times K}$ is a diagonal matrix whose $k$-th diagonal entries weights the contribution of device $k$ into the average achieved computation error\footnote{This point becomes clear as the reader goes through Section~\ref{sec:derive}.} in iteration $t$. The mapping $\Pi\brc{\cdot}$ moreover uses the contribution of device $k$ into the average computation error and updates its weight for the next iteration. This latter mapping is the main degree of freedom in Algorithm~\ref{alg:MPS}. By fixing the updating strategy $\Pi\brc{\cdot}$, one can run the algorithm by iterating for a few number of iterations till the algorithm converges to a feasible point. Noting that the algorithm in each iteration calculates only a single \ac{svd}, one can readily conclude that the proposed algorithm significantly reduces the complexity compared to the benchmark and close-to-optimal approaches illustrated in Example~\ref{ex:Alg}. We discuss the updating strategy and complexity of the algorithm in greater detail in the forthcoming subsections.

\subsection{Weight Updating via Subset-Cutting}
In general, the weight of device $k$, i.e., $\dbc{\mW\itr{t}}_{kk}$, in each iteration should be tuned such that the algorithm converges as fast as possible. Nevertheless, an extensive search for $\dbc{\mW\itr{t}}_{kk}$ leads to high complexity. To avoid this unnecessary complexity, we apply a simple binary tuning strategy to which we refer as \textit{subset-cutting}. For the sake of compactness, we define the notation $w_k\itr{t} = \dbc{\mW\itr{t}}_{kk}$ and assume that $0 < w_k\itr{t} <1$ for $k\in \setS\itr{t} $.

Let us cut the subset of selected devices in iteration $t$ into two subsets: 
\begin{itemize}
	\item  $\setS_+\itr{t}$ which contains those selected devices whose corresponding constraint indicator, i.e., $F_k\itr{t}$ defined in \eqref{eq:F_k}, is positive.
	\item $\setS_-\itr{t}$ which includes the remaining devices.
\end{itemize}
Considering the optimization in \eqref{eq:E_func}, our ultimate goal is to update the weights, such that the subset\footnote{The subset of devices in iteration $t+1$ whose corresponding computation error constraint is violated.} $\setS_+\itr{t+1}$ contains as few devices as possible. The simplest strategy to achieve this goal is to weight the two subsets inversely proportional. We apply this idea, by modeling the sign of $F_k\itr{t}$ for $k\in \setS$ as \ac{iid} Bernoulli random variables and letting the weight of each device be its probability. More precisely, we set the devices in $\setS_+\itr{t}$ and  $\setS_-\itr{t}$ to be weighted by $\delta$ and $1-\delta$ for some $\delta \in \brc{0,1}$, respectively. This means that the mapping $\Pi\brc\cdot$ in Algorithm~\ref{alg:MPS} is set to
\begin{align}
	\Pi \brc{ x } = \begin{cases}
		\delta & x>0\\
		1-\delta & x\leq 0
	\end{cases}
\end{align}
for some  of $0 < \delta <1$. The value of $\delta$ is then tuned for a given channel model, based on the statistics of the channel.

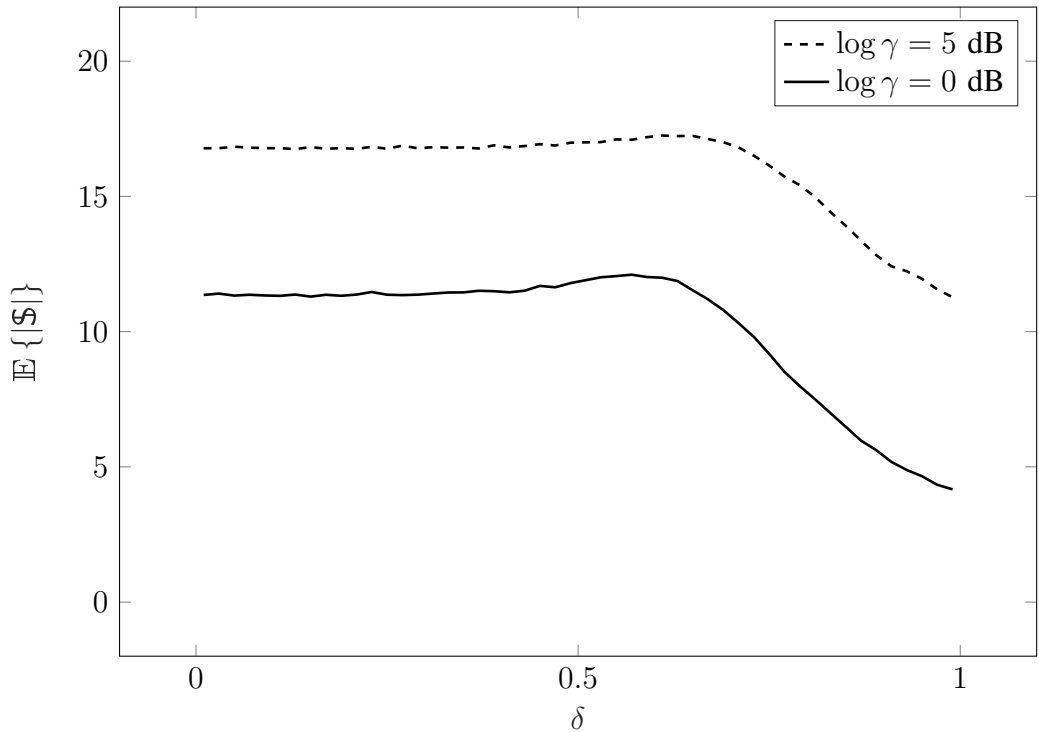
\begin{figure}
	\centering
%
%
\begin{tikzpicture}

\begin{axis}[%
width=4.8in,
height=3.4in,
at={(1.262in,0.697in)},
scale only axis,
xmin=-0.1,
xmax=1.1,
xtick={{0}, {.5} , {1}},
xticklabels={{$0$} , {$0.5$} , {$1$}},
xlabel style={font=\color{white!15!black}},
xlabel={$\delta$},
ymin=-2,
ymax=22,
ytick={{0}, {5} , {10}, {15}, {20}},
yticklabels={{$0$} , {$5$} , {$10$}, {$15$} , {$20$}},
ylabel style={font=\color{white!15!black}},
ylabel={$\Ex{\abs{\setS}}{ }$},
axis background/.style={fill=white},
legend style={legend cell align=left, align=left, draw=white!15!black}
]
\addplot [dashed, color=black, line width=1.0pt]
  table[row sep=crcr]{%
0.01	16.78\\
0.03	16.782\\
0.05	16.8375\\
0.07	16.8035\\
0.09	16.784\\
0.11	16.7835\\
0.13	16.749\\
0.15	16.825\\
0.17	16.7625\\
0.19	16.7835\\
0.21	16.761\\
0.23	16.834\\
0.25	16.761\\
0.27	16.8765\\
0.29	16.78\\
0.31	16.8195\\
0.33	16.8035\\
0.35	16.809\\
0.37	16.7715\\
0.39	16.8895\\
0.41	16.81\\
0.43	16.863\\
0.45	16.9305\\
0.47	16.882\\
0.49	16.985\\
0.51	16.9995\\
0.53	17.009\\
0.55	17.107\\
0.57	17.095\\
0.59	17.191\\
0.61	17.25\\
0.63	17.232\\
0.65	17.2365\\
0.67	17.1215\\
0.69	17.005\\
0.71	16.808\\
0.73	16.5065\\
0.75	16.1405\\
0.77	15.734\\
0.79	15.417\\
0.81	14.979\\
0.83	14.4275\\
0.85	13.9215\\
0.87	13.3665\\
0.89	12.829\\
0.91	12.416\\
0.93	12.233\\
0.95	11.95\\
0.97	11.5635\\
0.99	11.271\\
};
\addlegendentry{$\log \gamma = 5$ dB}

\addplot [color=black, line width=1.0pt]
table[row sep=crcr]{%
	0.01	11.357\\
	0.03	11.4065\\
	0.05	11.333\\
	0.07	11.364\\
	0.09	11.337\\
	0.11	11.324\\
	0.13	11.371\\
	0.15	11.294\\
	0.17	11.363\\
	0.19	11.3255\\
	0.21	11.37\\
	0.23	11.466\\
	0.25	11.365\\
	0.27	11.3505\\
	0.29	11.366\\
	0.31	11.4085\\
	0.33	11.447\\
	0.35	11.4505\\
	0.37	11.511\\
	0.39	11.495\\
	0.41	11.454\\
	0.43	11.514\\
	0.45	11.6915\\
	0.47	11.6395\\
	0.49	11.7935\\
	0.51	11.902\\
	0.53	12.0105\\
	0.55	12.053\\
	0.57	12.1065\\
	0.59	12.02\\
	0.61	11.993\\
	0.63	11.871\\
	0.65	11.528\\
	0.67	11.1935\\
	0.69	10.7985\\
	0.71	10.3145\\
	0.73	9.805\\
	0.75	9.177\\
	0.77	8.5145\\
	0.79	7.9855\\
	0.81	7.5005\\
	0.83	6.993\\
	0.85	6.486\\
	0.87	5.974\\
	0.89	5.6215\\
	0.91	5.182\\
	0.93	4.886\\
	0.95	4.653\\
	0.97	4.3355\\
	0.99	4.1685\\
};
\addlegendentry{$\log \gamma = 0$ dB}

\end{axis}
\end{tikzpicture}%
	\caption{Average $\abs{\setS}$ against $\delta$.}
	\label{fig:Nr_vs_delta}
\end{figure}

Numerical experiments show that the tuning of $\delta$ is rather robust. A particular case is observed in Fig.~\ref{fig:Nr_vs_delta} which shows the average $\abs{\setS}$ against $\delta$. Here, we set $K=20$ and $N=6$ and calculate the average numerically over $2000$ realizations of \ac{iid} complex Gaussian channels with zero mean and unit variance. The figure is given for two different values of the parameter $\gamma$. As we can see, the optimal choice of $\delta$ changes with $\gamma$. Nevertheless, the algorithm performs almost identically for \textit{small enough} choices of $\delta$. This indicates that in practice, it is enough to keep $\delta$ small, and an update of its value in each iteration is not necessary. Consequently, in the rest of this paper, we use this subset-cutting strategy with a fixed choice of $\delta$ to update the weight of the devices, and leave further discussions on more efficient update strategies as a direction for future work.

\subsection{Efficiency of Scheduling via Matching Pursuit}
The matching pursuit approach reduces the computational complexity of scheduling at the cost of performance degradation. We now apply some initial numerical investigations to quantify the scale of this complexity-performance trade-off. To this end, the same setting as the one in Fig.~\ref{fig:Nr_vs_delta} is considered, i.e., $N=6$, $K=20$ and $\phi_k = 1$ for $k\in\dbc{K}$. The uplink channels are further generated \ac{iid} Gaussian with zero mean and unit variance. For this setting, the matching pursuit approach, as well as the benchmark and close-to-optimal scheduling policies \cite{yang2020federated}, are used to perform scheduling for various choices of $\gamma$ and the performance is averaged over $2000$ channel realizations. The results are then represented in Fig.~\ref{fig:Nr_vs_gamma}, where the average number of selected devices is plotted against $\gamma$. 

As the figure shows, the proposed algorithm performs rather close to the close-to-optimal scheduling and outperforms considerably the state-of-the-art; see \cite{yang2020federated} for more details on the  close-to-optimal and the benchmark. It however is of a significantly lower computational complexity compared to the benchmark, as we show in the forthcoming subsection. This observation\footnote{Along with the complexity analysis in the next sub-section.} implies that the proposed approach based on matching pursuit is highly efficient for scheduling.

The step-wise nature of Algorithm~\ref{alg:MPS} brings this question into mind whether by using common extensions of step-wise regression techniques, the performance of this algorithm can be dominantly enhanced. To find an answer to this question, we consider the particular instance of extending Algorithm~\ref{alg:MPS} via a \textit{bidirectional} search. To this end, we note that the search in Algorithm~\ref{alg:MPS} is performed in a \textit{backward} fashion. A standard extension to this approach is to add one further step of \textit{forward selection} after the \textit{backward} search is over. By doing so, \textit{wrongly-removed} devices can be added again to the support. We hence extend the basic algorithm by revisiting the subset of removed devices once again at the end of each iteration, and place back to $\setS$ those removed devices whose updated constraint indicator is not positive. As observed in Fig.~\ref{fig:Nr_vs_gamma}, this extension is practically of no gain and only increases the complexity\footnote{Note that the corresponding algorithm still has significantly lower computational complexity compared to the benchmark and close-to-optimal policies.}. This implies that Algorithm~\ref{alg:MPS} is efficient from both performance and complexity viewpoints.
 

\begin{figure}
	\centering
%
%
\begin{tikzpicture}
	
	\begin{axis}[%
		width=4.8in,
		height=3.4in,
		at={(1.262in,0.697in)},
		scale only axis,
		xmin=-11,
		xmax=31,
		xlabel style={font=\color{white!15!black}},
		xlabel={$\gamma$},
		ymin=-2,
		ymax=22,
		ytick={{0}, {5} , {10}, {15}, {20}},
		yticklabels={{$0$} , {$5$} , {$10$}, {$15$} , {$20$}},
		ylabel style={font=\color{white!15!black}},
		ylabel={$\Ex{\abs{\setS}}{ }$},
		axis background/.style={fill=none},
		legend style={at={(.97,.29)},legend cell align=left, align=left, draw=white!15!black}
		]
		\addplot [ color=black, line width=1.0pt, mark=o, mark options={solid, black}]
		table[row sep=crcr]{%
-10	1\\
-7	2.2\\
-4	4.9\\
-1	9.5\\
2	13.8\\
5	17.5\\
8	19\\
11	19.6\\
14	20\\
17	20\\
20	20\\
23	20\\
		};
\addlegendentry{Algorithm~1}
		
\addplot [dashdotdotted,color=black, line width=1.0pt]
		table[row sep=crcr]{%
-10	1\\
-7	2.428\\
-4	5.912\\
-1	10.421\\
2	14.608\\
5	17.5\\
8	19\\
11	19.6\\
14	20\\
17	20\\
20	20\\
23	20\\
		};
\addlegendentry{Extension}
		
\addplot [ dashed, color=black, line width=.5pt, mark=square, mark options={solid, black}]
		table[row sep=crcr]{%
-10	1\\
-7	2.8\\
-4	6.9\\
-1	11.8\\
2	17.9\\
5	19.6\\
8	19.8\\
11	19.9\\
14	19.9\\
17	20\\
20	20\\
23	20\\
		};
\addlegendentry{Close-to-optimal}
		
\addplot [ dotted, color=black, line width=.5pt, mark=x, mark options={solid, black}]
		table[row sep=crcr]{%
			-10	0.1\\
			-5	1.25\\
			0	1.3\\
			5	2.5\\
			10	13\\
			15	19\\
			20	19.5\\
			25	19.7\\
			30	19.9\\
		};
		\addlegendentry{Benchmark}
		
	\end{axis}
\end{tikzpicture}%
	\caption{Average $\abs{\setS}$ against $\gamma$ for the proposed algorithm as well as the benchmark and close-to-optimal approach.}
	\label{fig:Nr_vs_gamma}
\end{figure}
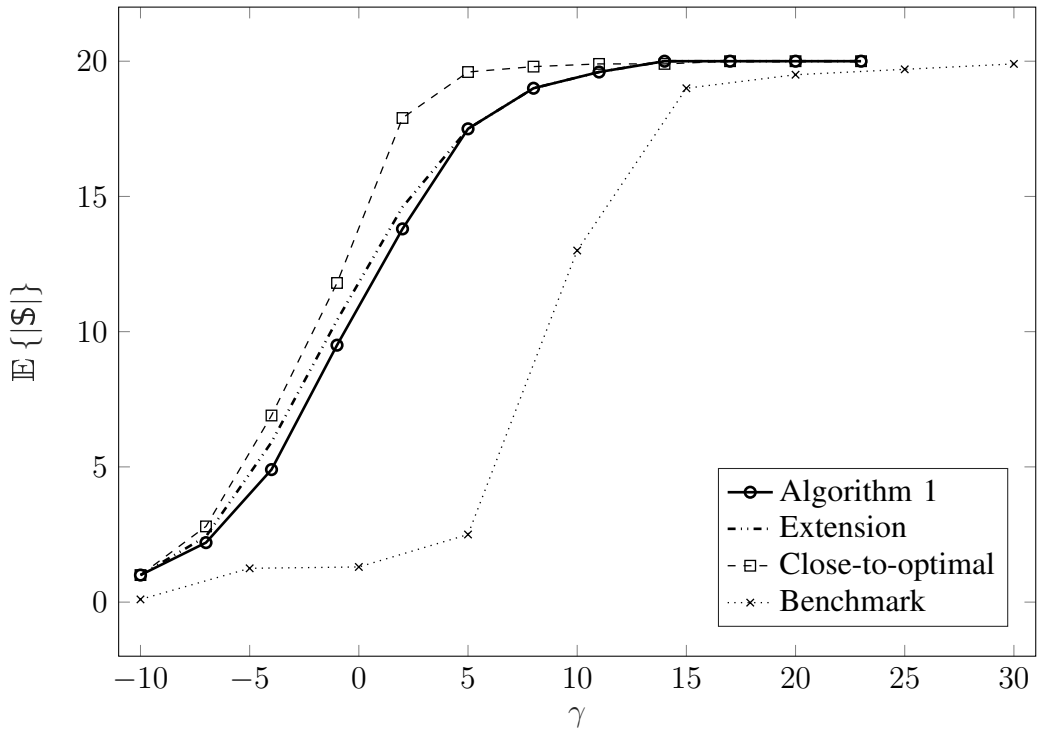

\subsection{Complexity of Scheduling via Matching Pursuit}
The key property of Algorithm~\ref{alg:MPS} is its significant low-complexity. We illustrate this point by comparing its complexity to that of the benchmark and close-to-optimal scheduling policies. The benchmark approach solves the scheduling problem iteratively with a fixed number of iterations; see the algorithms based on semidefinite relaxation in \cite{yang2020federated}. In each iteration, the algorithm performs two steps of updates, where the first step is of the complexity order $\mao\brc{\brc{N^2+K}^3}$, and the second step is of the complexity order $\mao\brc{N^6}$. We can hence conclude that the benchmark complexity for scheduling is
\begin{align}
	\mac_{\rm B} &= \mao\brc{\brc{N^2+K}^3 + N^6}
\end{align}

The close-to-optimal approach that uses \ac{dc} programming also runs iteratively a two-step update with the same order of complexity. It runs the update in each iteration for all available $K$ devices; see \cite{yang2020federated}. As a result, the complexity of the close-to-optimal approach is
\begin{align}
	\mac_{\rm DCP} &= \mao\brc{K\brc{N^2+K}^3 + KN^6}.
\end{align}

We now consider the proposed algorithm based on matching pursuit. This algorithm runs for $T \leq K$ iterations. In iteration $t$, an \ac{svd} decomposition of a matrix of size $N\times \brc{K-t+1}$ is required\footnote{Remember that $\mW\itr{t}$ has $t-1$ all-zero columns.} whose complexity is 
\begin{align}
	\mac_{\rm SVD}\itr{t} = \mao\brc{\brc{K-t+1}N^2} \leq \mao\brc{ K N^2 }.
\end{align}
The algorithm is hence of the complexity order of 
\begin{align}
	\mac_{\rm MP} = \mao\brc{\brc{K-T}KN^2} \leq \mao\brc{ K^2 N^2 }
\end{align}
which is drastically less than the complexity of the benchmark and close-to-optimal policies. It is worth mentioning that the complexity of Algorithm~\ref{alg:MPS} can be even further reduced. In fact, in each iteration, we only need the singular-vector which corresponds to the largest singular-value of the weighted channel matrix; see line 5 in the algorithm. A complete \ac{svd} determination is hence unnecessary and one can use alternative algorithms to find only the desired singular-vector. An example is the Lanczos algorithm whose complexity, though not explicitly derived, is known to be smaller than \ac{svd} calculation \cite{parlett1982estimating}. Further discussions on such algorithms can be found in \cite{SCHWETLICK20031,liang2014computing}. One can hence conclude that the complexity of the proposed algorithm is in general of the order of
\begin{align}
	\mac_{\rm MP} &= \mao\brc{K^p N^q}
\end{align}
for some $0< p,q<2$. This is a huge complexity reduction at the expense of a minor performance  degradation which as we see later in Section~\ref{sec:NumRes} can  be neglected in most practical scenarios.

Fig.~\ref{fig:Time_vs_gamma} compares the average run-time of Algorithm~\ref{alg:MPS} with the one of the close-to-optimal approach for the setting considered in Fig.~\ref{fig:Nr_vs_gamma}. Here, the time axis is shown in logarithmic scale, in order to enable comparison.  The results clearly show the drastic complexity reduction achieved by the matching pursuit approach. These results along with those represented in the previous subsection imply that the proposed algorithm is a suitable candidate for scheduling in real-time applications of \ac{otaFL}.

\begin{remark}
	As Fig.~\ref{fig:Time_vs_gamma}, the complexity of both algorithms drops considerably as $\gamma$, i.e., the error tolerance level, increases. This is due to the fact that both algorithms select more devices at higher choices of $\gamma$, and hence terminate faster in this regime. Considering the step-wise nature of Algorithm~\ref{alg:MPS}, we can further flatten the complexity curve by switching from \textit{backward selection} to \textit{forward selection} at lower tolerances. More precisely, for small choices of $\gamma$, one can start from the empty set $\setS = \emptyset$ and add active devices in each iteration. The derivations for this complementary algorithm follows the exact steps as in Algorithm~\ref{alg:MPS}. We hence skip it at this point and leave it as a natural extension of this work. 
\end{remark}
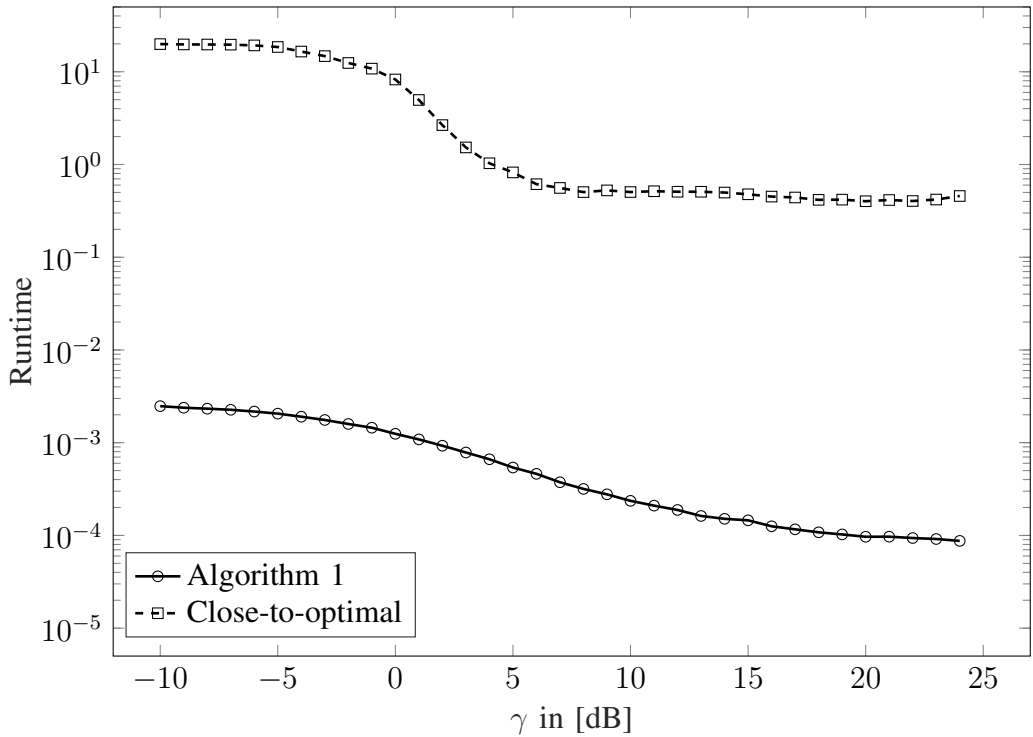
\begin{figure}
	\centering
%
%
\begin{tikzpicture}
	
	\begin{axis}[%
		width=4.8in,
		height=3.4in,
		at={(1.262in,0.697in)},
		scale only axis,
		xmin=-12,
		xmax=27,
		xlabel style={font=\color{white!15!black}},
		xtick={-10,-5,0,5,10,15,20,25},
		xticklabels={{$-10$},{$-5$},{$0$},{$5$},{$10$},{$15$},{$20$},{$25$}},
		xlabel={$\gamma$ in [dB]},
		ymin=0.000005,
		ymax=50,
		ymode=log,
		ytick={0.00001, 0.0001, 0.001,.01,.1,1,10,100},
		yticklabels={{$10^{-5}$},{$10^{-4}$},{$10^{-3}$},{$10^{-2}$},{$10^{-1}$},{$10^{0}$},{$10^{1}$},{$10^{2}$}},
		ylabel style={font=\color{white!15!black}},
		ylabel={Runtime},
		axis background/.style={fill=none},
		legend style={at={(.33,.16)},legend cell align=left, align=left, draw=white!15!black}
		]

\addplot [color=black, line width=1.0pt, mark=o, mark options={thin, black}]
table[row sep=crcr]{%
-10	0.00247808700000\\
-9	0.00238302700000\\
-8	0.00232718000000\\
-7	0.00226460000000\\
-6	0.00216760900000\\
-5	0.00205823200000\\
-4	0.00190670500000\\
-3	0.00175786000000\\
-2	0.00159120600000\\
-1	0.00145067500000\\
0	0.00124719100000\\
1	0.00108513800000\\
2	0.00092774870000\\
3	0.00078405620000\\
4	0.00066284180000\\
5	0.00053996090000\\
6	0.00046159270000\\
7	0.00037439580000\\
8	0.00031739470000\\
9	0.00027736900000\\
10	0.00023603200000\\
11	0.00020959620000\\
12	0.00018816470000\\
13	0.00016218900000\\
14	0.00015103580000\\
15	0.00014548060000\\
16	0.00012548690000\\
17	0.00011634110000\\
18	0.00010802030000\\
19	0.00010248900000\\
20	0.00009672642000\\
21	0.00009690762000\\
22	0.00009373665000\\
23	0.00009149313000\\
24	0.00008729935000\\
};
\addlegendentry{Algorithm~1}

\addplot [dashed,color=black, line width=1.0pt, mark=square, mark options={solid,thin, black}]
table[row sep=crcr]{%
-10	19.884981\\
-9	19.726808\\
-8	19.683913\\
-7	19.622551\\
-6	19.254122\\
-5	18.506432\\
-4	16.525658\\
-3	14.756451\\
-2	12.443654\\
-1	10.811022\\
0	8.255982\\
1	4.963903\\
2	2.660218\\
3	1.528051\\
4	1.031238\\
5	0.821974\\
6	0.613569\\
7	0.558712\\
8	0.504097\\
9	0.525794\\
10	0.504733\\
11	0.514993\\
12	0.508637\\
13	0.508222\\
14	0.498527\\
15	0.477361\\
16	0.45102\\
17	0.44091\\
18	0.416036\\
19	0.417889\\
20	0.401902\\
21	0.413939\\
22	0.403971\\
23	0.419039\\
24	0.457426\\
};
\addlegendentry{Close-to-optimal}

%
		
	\end{axis}
\end{tikzpicture}%
	\caption{Runtime in seconds against $\gamma$.}
	\label{fig:Time_vs_gamma}
\end{figure}

\section{Derivation of the Algorithm}
\label{sec:derive}
In this section, we give the detailed derivations for Algorithm~\ref{alg:MPS}. Starting with an initialization, let the selected support and the linear receiver updated in iteration $t$ be denoted by $\setS\itr{t}$ and $\bmm\itr{t}$, respectively. Given that $\setS\itr{t}$ and $\bmm\itr{t}$ do not lead to a feasible point for \eqref{eq:E_func}, in this iteration, i.e., iteration $t$, we intend to select a device to be eliminated from $\setS\itr{t}$ and update the linear receiver.

Let us denote the index of the device that has been selected at the end of iteration $t-1$ for this iteration by $i\itr{t}$. Considering Algorithm~\ref{alg:MPS}, at the beginning of the iteration, we update the support by removing device $i\itr{t}$, i.e., 
\begin{align}
	\setS\itr{t} = \setS\itr{t-1} - \set{i\itr{t}}.
\end{align}
In the sequel, we derive the update rule for the linear receiver $\bmm\itr{t}$ and the new index $i\itr{t+1}$.

\subsection{Updating the Linear Receiver}
We first update the linear receiver for the given support, i.e., we find $\bmm\itr{t}$. To this end, let us define the constraint function for device $k\in\setS\itr{t}$ as
\begin{align}
	f_k\brc{\bmm} = \phi_k^2 {\norm{\bmm }^2} - \gamma {\abs{\bh_k^\her \bmm }^2}.
\end{align}
We desire that for all $k\in\setS\itr{t}$, we have $f_k\brc{\bmm} \leq 0 $. A necessary, but \textit{not sufficient}, condition for this objective is that a weighted sum of these constraint functions with positive coefficients be negative. We use this fact, and define an \textit{average} constraint function in iteration $t$ as
\begin{align}
	F\itr{t}\brc{\bmm} = \sum_{k \in \setS\itr{t} } w_k\itr{t} f_k\brc{\bmm} 
\end{align}
for some $w_k\itr{t}  > 0 $ for all $k \in \setS\brc{t}$.

From a probabilistic viewpoint, a desired behavior is to have $F\itr{t}\brc{\bmm}$ be as negative as possible\footnote{Note that these statements only follow from necessary conditions and hence are \textit{heuristic}.}. We hence update $\bmm\itr{t}$ as
\begin{align}
	\bmm\itr{t} &= \argmin_{\bmm} \; F\itr{t}\brc{\bmm} \label{eq:c1}\\
	&=  \argmin_{\bmm} \;  \sum_{k \in \setS\itr{t} } {w_k\itr{t}} \brc{\phi_k^2 {\norm{\bmm }^2}  - \gamma {\abs{\bh_k^\her \bmm }^2}}\\
	&=  \argmin_{\bmm} \; \bmm^\her \brc{\Phi\itr{t} \mI_N  - \gamma \mH \mW\itr{t} \mH^\her } \bmm 
\end{align}
where we define the weighted average of global weights as
\begin{align}
	\Phi\itr{t}  =  \sum_{k \in \setS\itr{t} } w_k\itr{t} \phi_k^2,
\end{align}
and the diagonal matrix $\mW\itr{t} \in \setR_+^{K\times K}$ as
\begin{align}
	\mW\itr{t} = \begin{cases}
		0 & k \notin  \setS\brc{t}\\
		w_k\itr{t} & k \in  \setS\brc{t}
	\end{cases}.
\end{align}

The solution to this latter optimization is readily given by the \ac{svd} considering the known bounds on the Rayleigh quotient \cite{horn2012matrix}. In fact, $F\itr{t}\brc{\bmm}$ is simply minimized by setting $\bmm$ to be the eigenvector corresponding to the minimum eigenvalue of $\Phi\itr{t} \mI_N  - \gamma \mH \mW\itr{t} \mH^\her$. Let us denote the \ac{svd} of $\mH \sqrt{\mW\itr{t}}$ as
\begin{align}
	\mH \sqrt{\mW\itr{t}} = \mV\itr{t} \mSigma\itr{t} \mU^{\brc{t} \her}
\end{align}
where $\mV\itr{t}\in \setC^{N\times N}$ and $\mU\itr{t}\in\setC^{K\times K}$ are unitary matrices and $\mSigma\itr{t}\in \setR_+^{N\times K}$ denotes the matrix of singular values, i.e., 
\begin{align}
	\mSigma\itr{t} = \dbc{\Diag{\sqrt{\sigma_1\itr{t}}, \ldots,\sqrt{ \sigma_N\itr{t}} } \left\vert \boldsymbol{0}_{N\times \brc{N-K}} \right. }
\end{align}
with $\sigma_n\itr{t}$ denoting the squared singular value of $\mH \sqrt{\mW\itr{t}}$. Define the index of the maximum squared singular value of $\mH \sqrt{\mW\itr{t}}$ as
\begin{align}
	n\itr{t} = \argmax_{n\in\dbc{N} } \; \sigma_n\itr{t}
\end{align}
and let the maximum value be 
\begin{align}
	\rho\itr{t} = \max_{n\in\dbc{N} } \; \sigma_{n} \itr{t}.
\end{align}
Moreover, let $\bvv\itr{t}$ and $\buu\itr{t}$ be the $n\itr{t}$-th columns of $\mV\itr{t}$ and $\mU\itr{t}$, respectively. We then have
\begin{align}
	\min_{\bmm} \; \bmm^\her \brc{\Phi\itr{t} \mI_N  - \gamma \mH \mW\itr{t} \mH^\her } \bmm  = \Phi\itr{t} - \gamma \rho\itr{t}
\end{align}
which is given by setting
\begin{align}
	\bmm =  \bvv\itr{t}. \label{eq:c2}
\end{align}
We therefore update the linear receiver $\bmm\itr{t}$ as $\bmm\itr{t} = \bvv\itr{t}$.

\subsection{Finding the New Index}
After updating the linear receiver, we check whether the updated $\bmm\itr{t}$ results in satisfied constraints for all the devices or not. To this end, we determine the following maximal constraint function over the selected subset of devices: 
\begin{align}
	\Delta\itr{t} = \max_{k \in \setS\itr{t} } f_k\brc{\bmm\itr{t}}.
\end{align}
Noting that $\norm{\bmm\itr{t}}^2 = 1$, we can conclude that
\begin{align}
	\Delta\itr{t} = \max_{k \in \setS\itr{t} } \phi_k^2  - \gamma {\abs{\bh_k^\her \bmm\itr{t} }^2}.
\end{align}
This metric decides whether further iterations are needed or not. Namely, when $\Delta\itr{t} \leq 0$, we can conclude that the computation error lies below the maximum tolerable limit and hence the algorithm will stop. With $\Delta\itr{t} > 0$, the constraint on the computation error is not yet satisfied and a further iteration is required. 

In the case of $\Delta\itr{t} > 0$, we start the next iteration by removing device $i\itr{t+1}$, where
\begin{align}
	i\itr{t+1} &= \argmax_{k \in \setS\itr{t} } f_k\brc{\bmm\itr{t}}\\
	&= \argmax_{k \in \setS\itr{t} } \phi_k^2  - \gamma {\abs{\bh_k^\her \bmm\itr{t} }^2}\\
	&= \argmin_{k \in \setS\itr{t} } {\abs{\bh_k^\her \bmm\itr{t} }^2}.
\end{align}
From the update rule, one can observe that the selected device is the one whose relative inner product of its uplink channel with the current receiver is \textit{minimal}. This describes in fact a matching pursuit strategy, and hence is the reason behind the appellation. 

\section{Numerical Experiments}
\label{sec:NumRes}
In this section, we investigate the performance of the proposed algorithm through several numerical experiments. In this respect, we consider a simple multiuser network in which multiple edge devices are to learn the classification of images from a shared dataset via \ac{FL}. Using the standard dataset  released by the Canadian Institute for Advanced Research \cite{krizhevsky2009learning}, known as CIFAR-10, we investigate the performance of the proposed algorithm, and compare it with the benchmark, as well as the close-to-optimal algorithm. 

\subsection{Network Setting}
We consider a typical \ac{otaFL} setting consisting of a \ac{ps} and $K$ single-antenna wireless devices. The \ac{ps} is equipped with $N$ antenna elements arranged on a \ac{ULA}. We use a two-dimensional Cartesian coordinate system to describe the layout of the network\footnote{This is a rather accurate approximation assuming that the heights of the edge devices change within a relatively small range compared to the height of the \ac{ps}.}, which is depicted in Fig.~\ref{fig:layout}. The \ac{ps} is located at the origin while the wireless devices are uniformly placed within a ring with inner radius $R_\text{in}$ and outer radius $R_\text{out}$ around the \ac{ps}.

\begin{figure}[t]
	\centering
	\begin{tikzpicture}
\tikzset{mobile phone/.pic={
code={
\begin{scope}[line join=round,looseness=0.25, line cap=round,scale=0.07, every node/.style={scale=0.07}]
\begin{scope}
\clip [preaction={left color=blue!10, right color=blue!30}] 
  (1/2,-1) to [bend left] (0,10)
  to [bend left] ++(1,1) -- ++(0,2)
  arc (180:0:3/4 and 1/2) -- ++(0,-2)
  to [bend left]  ++(5,-2) coordinate (A) to [bend left] ++(-1/2,-11)
  to [bend left] ++(-1,-1) to [bend left] cycle;
\path [left color=blue!30, right color=blue!50]
  (A) to [bend left] ++(0,-11) to[bend left] ++(-3/2,-2)
  -- ++(0,12);
\path [fill=blue!20, draw=white, line width=0.01cm]
  (0,10) to [bend left] ++(1,1) -- ++(0,2)
  arc (180:0:3/4 and 1/2) -- ++(0,-2)
  to [bend left]  (A) to [bend left] ++(-3/2,-5/4)
  to [bend right] cycle;
\draw [line width=0.01cm, fill=white]
  (9/8,21/2) arc (180:360:5/8 and 3/8) --
  ++(0,2.5) arc (0:180:5/8 and 3/8) -- cycle;
\draw [line width=0.01cm, fill=white]
  (9/8,13) arc (180:360:5/8 and 3/8);
\fill [white, shift=(225:0.5)] 
  (1,17/2) to [bend left] ++(4,-7/4)
  to [bend left] ++(0,-7/2) to [bend left] ++(-4, 6/4)
  to [bend left] cycle;
\fill [black, shift=(225:0.25)] 
  (1,17/2) to [bend left] ++(4,-7/4)
  to [bend left] ++(0,-7/2) to [bend left] ++(-4, 6/4)
  to [bend left] cycle;
\shade [inner color=white, outer color=cyan!20] 
  (1,17/2) to [bend left] ++(4,-7/4)
  to [bend left] ++(0,-7/2) to [bend left] ++(-4, 6/4)
  to [bend left] cycle;
%
\end{scope}
\draw [line width=0.02cm] 
  (1/2,-1) to [bend left] (0,10)
  to [bend left] ++(1,1) -- ++(0,2)
  arc (180:0:3/4 and 1/2) -- ++(0,-2)
  to [bend left]  ++(5,-2) to [bend left] ++(-1/2,-11)
  to [bend left] ++(-1,-1) to [bend left] cycle;
\end{scope}%
}}}

\tikzset{eve/.pic={
		code={
			\begin{scope}[line join=round,looseness=0.25, line cap=round,scale=0.07, every node/.style={scale=0.07}]
				\begin{scope}
					\clip [preaction={left color=red!10, right color=red!30}] 
					(1/2,-1) to [bend left] (0,10)
					to [bend left] ++(1,1) -- ++(0,2)
					arc (180:0:3/4 and 1/2) -- ++(0,-2)
					to [bend left]  ++(5,-2) coordinate (A) to [bend left] ++(-1/2,-11)
					to [bend left] ++(-1,-1) to [bend left] cycle;
					\path [left color=red!30, right color=red!50]
					(A) to [bend left] ++(0,-11) to[bend left] ++(-3/2,-2)
					-- ++(0,12);
					\path [fill=red!20, draw=white, line width=0.01cm]
					(0,10) to [bend left] ++(1,1) -- ++(0,2)
					arc (180:0:3/4 and 1/2) -- ++(0,-2)
					to [bend left]  (A) to [bend left] ++(-3/2,-5/4)
					to [bend right] cycle;
					\draw [line width=0.01cm, fill=white]
					(9/8,21/2) arc (180:360:5/8 and 3/8) --
					++(0,2.5) arc (0:180:5/8 and 3/8) -- cycle;
					\draw [line width=0.01cm, fill=white]
					(9/8,13) arc (180:360:5/8 and 3/8);
					\fill [white, shift=(225:0.5)] 
					(1,17/2) to [bend left] ++(4,-7/4)
					to [bend left] ++(0,-7/2) to [bend left] ++(-4, 6/4)
					to [bend left] cycle;
					\fill [black, shift=(225:0.25)] 
					(1,17/2) to [bend left] ++(4,-7/4)
					to [bend left] ++(0,-7/2) to [bend left] ++(-4, 6/4)
					to [bend left] cycle;
					\shade [inner color=white, outer color=orange!20] 
					(1,17/2) to [bend left] ++(4,-7/4)
					to [bend left] ++(0,-7/2) to [bend left] ++(-4, 6/4)
					to [bend left] cycle;
					%
				\end{scope}
				\draw [line width=0.02cm] 
				(1/2,-1) to [bend left] (0,10)
				to [bend left] ++(1,1) -- ++(0,2)
				arc (180:0:3/4 and 1/2) -- ++(0,-2)
				to [bend left]  ++(5,-2) to [bend left] ++(-1/2,-11)
				to [bend left] ++(-1,-1) to [bend left] cycle;
			\end{scope}%
}}}

\tikzset{radiation/.style={{decorate,decoration={expanding waves,angle=90,segment length=4pt}}},
	antenna/.pic={
		code={\tikzset{scale=3/10}
			\draw[semithick] (0,0) -- (1,4);
			\draw[semithick] (3,0) -- (2,4);
			\draw[semithick] (0,0) arc (180:0:1.5 and -0.5);
			\node[inner sep=4pt] (circ) at (1.5,5.5) {};
			\draw[semithick] (1.5,5.5) circle(8pt);
			\draw[semithick] (1.5,5.5cm-8pt) -- (1.5,4);
			\draw[semithick] (1.5,4) ellipse (0.5 and 0.166);
			\draw[semithick,radiation,decoration={angle=45}] (1.5cm+8pt,5.5) -- +(0:2);
			\draw[semithick,radiation,decoration={angle=45}] (1.5cm-8pt,5.5) -- +(180:2);
	}}
}

	\path (0,-.8) pic {antenna};
	\node (ps) at (0.5,-1.2) {PS};
	\draw[dashed] (0.5,0.7) circle[radius=3cm];
	\draw[dashed] (0.5,0.7) circle[radius=5.8cm];
	
	\path (3,-2.8) pic {mobile phone};
	\path (-3.1,-2.7) pic {mobile phone};
	\path (-3.2,3.7) pic {mobile phone};
	\path (4,2.1) pic {mobile phone};
	\path (2.5,4) pic {mobile phone};
	\path (-3.8,1.1) pic {mobile phone};
	\node (device) at (-2.8,-3.15) {\textcolor{blue!90}{Device $k$}};

	\draw [->,>=stealth,dashed] (.42,0.94) -- (-.5,3.5) node [above, sloped,pos=.7] (d) {$R_{\rm in}$};
	\draw [->,>=stealth,dashed] (.49,0.94) -- (1.5,6.35) node [above, sloped,pos=.7] (d) {$R_{\rm out}$};
	\draw [->,>=stealth,thick] (-2.4,-2) -- (-.25,.45) node [above, sloped,pos=.5] (d) {$\bh_k$};

\end{tikzpicture}
	\caption{Layout of the \ac{FL} system with one \ac{ps} and $K$ randomly distributed wireless devices. }
	\label{fig:layout}
\end{figure}
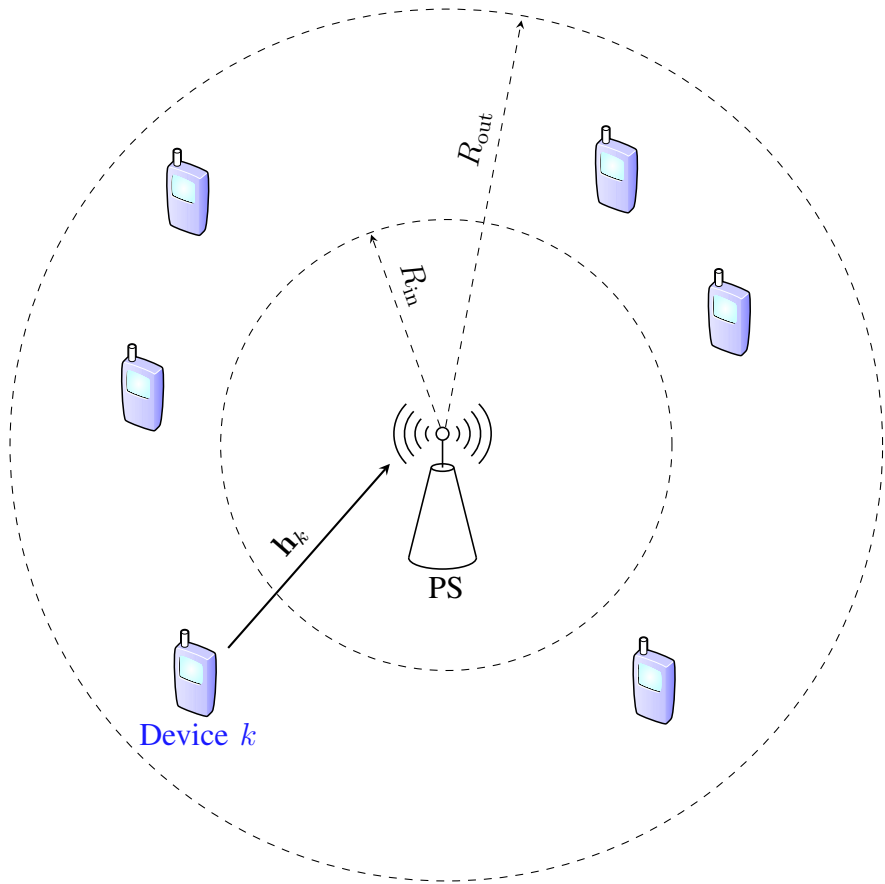

The channel model accounts for both large-scale and small-scale fading. Assuming that a direct \ac{los} path exists between the \ac{ps} and the wireless devices, a Rician distribution is used to model the small-scale fading. More specifically, the vector of uplink channel coefficients between device $k$ and the \ac{ps} is described by
\begin{align}
	\bh_k = \sqrt{\text{PL}(l_k)} \bgg_k,
\end{align}
where $\text{PL}(l_k)$ denotes the large-scale path-loss of device $k$ and is a function of its distance from the \ac{ps}, i.e., $l_k$, and $\bgg_k \in \setC^{N }$ describes the small-scale fading process in the channel and is given by
\begin{align}
	\bgg_k = \sqrt{\frac{\kappa_k}{1+\kappa_k}} \bar{\bgg}_k + \sqrt{\frac{1}{1+\kappa_k}} \tilde{\bgg}_k. \label{channel:small}
\end{align} 
In \eqref{channel:small}, the non-negative scalar $\kappa_k$ is the Rician factor, and $\bar{\bgg}_k $ and $\tilde{\bgg}_k $ denote the \ac{los} and \ac{nlos} component, respectively. Considering isotropic radiation in a rich scattering environment, the \ac{los} term can be written in terms of the specular array response at the \ac{ps} by \cite[Chapter~1]{SIG-093}
\begin{align}
	\bar{\bgg}_k = \dbc{ 1 , u\brc{\theta_{k}} , u^2\brc{\theta_{k}}  \dots , u^{N-1}\brc{\theta_{k}} }^\trp, 
\end{align} 
where $u\brc{\theta_{k}}$ is defined as
\begin{align}
	u\brc{\theta_{k}} = \exp\set{j2\pi d \sin\brc{\theta_k}}
\end{align}
with $\theta_k$ being the (azimuth) \ac{AoA} at the \ac{ps} from the $k$-th device, and $d$ being the distance between two neighboring antenna elements on the \ac{ULA} normalized by the wavelength. 

The \ac{nlos} component follows the Rayleigh fading model: $\bar{\bgg}_k $ is modeled by an \ac{iid} complex Gaussian process with zero mean and covariance matrix $\mR_k \in \setC^{N\times N}$, i.e., $\tilde{\bgg}_k \sim \mathcal{CN}\brc{\boldsymbol{0}, \mR_k}$. Assuming the number of impinging plane waves superposed at the \ac{ps} to be large, we set the elements of the covariance matrix to be \cite[Chapter~2]{SIG-093}
\begin{align}
	\dbc{\mR_k}_{n,m} = u^{n-m}\brc{\theta_{k}} \varrho_{n,m}^{\brc{k}}\brc{\theta_{k}},
\end{align}
where $\varrho_{n,m}\brc{\theta_{k}}$ accounts for the angular spread of the \ac{AoA} at the \ac{ps} from device $k$, and is given by
\begin{align}
	\varrho_{n,m}^{\brc{k}} \brc{\theta_{k}} = \exp\set{-2\varsigma_k^2 \dbc{\pi \brc{n-m} d \cos\brc{\theta_k}}^2}
\end{align}
with $\varsigma_k$ denoting the angular standard deviation of device $k$. Throughout the simulations, the nominal \acp{AoA} of the devices are determined geometrically from the randomly-generated location of the devices. The angular standard deviation of each device is further chosen uniformly at random from the interval $\dbc{12,15}$.

 The distance-dependent path loss model is further given by 
\begin{align}
	\text{PL}\brc{l_k} = \text{PL}_0 \left( \frac{l_k}{l_0}\right)^{-\alpha}, 
\end{align}
where $l_k$ denotes the distance between the \ac{ps} and the $k$-th user, $\text{PL}_0$ denotes the path loss value at the reference distance $l_0$, and $\alpha$ is the path loss exponent. Without loss of generality, throughout the simulations, we set $l_0 = l_{\min}$, i.e., the distance of the closest device to the \ac{ps}, and $\text{PL}_0 = 1$.

Throughout the  simulations, we set the numerical values of the communication network parameters as presented in Table \ref{tab:simulation_params}. 
\begin{table}[t]
	\normalsize
	\centering
	\begin{tabular}{| c | c | } 
		\hline
		Parameter & Value  \\ [0.5ex] 
		\hline\hline
		Inner radius $R_\text{in}$ in meters & 10  \\
		\hline
		Outer radius $R_\text{out}$ in meters & 100  \\
		\hline
		$\log \kappa_k$ in [dB] & $3$ \\
		\hline
		Antenna elements spacing $d$ (per wavelength) & $0.5$ \\
		\hline
		Path loss exponent $\alpha$ & 3 \\
		\hline
	\end{tabular}
\caption{Simulation parameters for the communication network.}
	\label{tab:simulation_params}
\end{table}

\subsection{Learning Setup}
We consider a $10$-class image classification task on the standard CIFAR-10 dataset. This dataset contains $60 000$ images, from which $L = 50 000$ are used as the training data and the remaining $10 000$ images as the test data. The overall dataset is partitioned into $K$ subsets and shared among the edge devices. 
The training data points are \textit{heterogeneously} distributed among the edge devices according to the following asymmetric approach: for a randomly-selected half of devices, $\tilde{L}_k$ is uniformly drawn from the interval  
		\begin{align}
		\dbc{L/K  , L/K + \epsilon_1}.
	\end{align} 
	for some non-negative real-valued $\epsilon_1$. For the remaining half, $\tilde{L}_k$ is uniformly drawn from the interval $\dbc{\epsilon_0, \epsilon_1}$ for some $\epsilon_0 < \epsilon_1$. The number of training images at device $k$ is then set to
	\begin{align}
		L_k = \tilde{L}_k + \left\lfloor \frac{1}{K} \brc{L-\sum_{k=1}^K \tilde{L}_k } \right\rfloor
	\end{align}
	This allows for a high variation among the local datasets sizes, i.e., quantity skew. The local datasets are further collected by an asymmetric sampling scheme, such that we observe label skew across the local datasets \cite{li2022federated}.


To learn the image classification task, a \ac{CNN} is locally trained at each of the selected wireless devices. The employed \ac{CNN} is designed as a simplified version of the widely known VGG13 network\footnote{Named after the visual geometry group at University of Oxford.} \cite{simonyan2015very}. Specifically, its architecture consists of eight convolutional blocks and two subsequent fully-connected layers. A convolutional block refers to the collection of a convolutional layer, a batch-normalization layer, and a rectified linear unit activation function. The convolutional filter size is fixed to $3 \times 3$, whereas the number of filters for the eight convolutional layers are set to $32$, $32$, $64$, $64$, $128$, $128$, $256$, and $256$, respectively. Moreover, a max-pooling operation with a filter size of $2$ and stride of $2$ is performed after each two convolutional blocks. Finally, we set the number of output nodes for the first fully-connected layer to $512$, which is followed by a dropout layer in order to avoid overfitting during training. All local \ac{CNN} models are trained using the \ac{sgd} with momentum algorithm to minimize a cross-entropy loss function. The training starts with an initial learning rate of $0.005$. The learning rate is eventually adjusted according to a \textit{reduce on plateau} strategy, which halves the learning rate if the test loss has not been improving over the last ten epochs. 

Throughout the simulations of the \ac{FL} process, we consider a time-varying Rician fading channel. We assume that the channel changes from one communication round to another, but it remains constant during the learning epochs within one round. We set the number of epochs per round to $10$ and the total number of communication rounds to $T = 6$. The scalars $\epsilon_0$ and $\epsilon_1$ are further set to $\epsilon_0 = 300$ and $\epsilon_1 = 500$. Other relevant parameters for training the CNN image classifier are shown in Table \ref{tab:learning_params}. 

\begin{table}[t]
	\normalsize
	\centering
	\begin{tabular}{| c | c |} 
		\hline
		Parameter & Value  \\ [0.5ex] 
		\hline\hline
		Number of train images $D$ & 50 000  \\ 
		\hline
		Number of test images & 10 000  \\
		\hline
		Number of epochs  & 60 \\
		\hline
		Number of communication rounds $T$ & 6  \\
		\hline
		Dropout probability & 0.5  \\
		\hline
		Learning rate $\beta$ & 0.005 \\
		\hline
		Batch size & 32 \\
		\hline
		Momentum & 0.9 \\
		\hline
	\end{tabular}
	\caption{Simulation parameters for the learning setting.}
	\label{tab:learning_params}
\end{table}

\subsection{Simulation Results}
We start the numerical investigations by giving the learning performance figures for Algorithm~1, as well as the closed-to-optimal and benchmark policies, in Figs.~\ref{fig:Training_Loss_vs_Epoch}, \ref{fig:Loss_vs_Epoch} and \ref{fig:Accuracy_vs_Epoch}. For these figures, we set the number of devices to $K=20$ and the size of the \ac{ps} \ac{ULA} to $N=6$. Fig.~\ref{fig:Training_Loss_vs_Epoch} shows the training loss versus the number of epochs averaged over multiple channel realizations considering both approaches. The test loss over the test dataset is further shown in Fig.~\ref{fig:Loss_vs_Epoch}. As a reference, we further plot the test loss achieved by \textit{perfect federated averaging}. The simulation results for this case is denoted as \textit{perfect \ac{FL}} in the figures representing the scenario in which all the devices are participating in the learning and communicate over a noiseless network. Perfect \ac{FL} hence gives a lower bound for training and test loss, and an upper bound for learning accuracy. The remaining parameters in the setting are selected from Tables~\ref{tab:simulation_params} and \ref{tab:learning_params} and the internal parameters of the algorithms are numerically optimized.

\begin{figure}
	\centering
%
%
\begin{tikzpicture}
	
	\begin{axis}[%
		width=4.8in,
		height=3.4in,
		at={(1.262in,0.697in)},
		scale only axis,
		xmin=0,
		xmax=61,
		xlabel style={font=\color{white!15!black}},
		xtick={10,20,30,40,50,60},
		xticklabels={{$10$},{$20$},{$30$},{$40$},{$50$},{$60$}},
		xlabel={epoch},
		ymin=0.02,
		ymax=2.45,
		ytick={0.5,1,1.5,2},
		yticklabels={{$0.5$},{$1$},{$1.5$},{$2$}},
		ylabel style={font=\color{white!15!black}},
		ylabel={Training Loss},
		axis background/.style={fill=none},
		legend style={legend cell align=left, align=left, draw=white!15!black}
		]

\addplot [color=black, line width=1.0pt, mark=o, mark options={thin, black}]
table[row sep=crcr]{%
1	2.32230086422722\\
2	1.63070321744675\\
3	1.42054657902625\\
4	1.307955432065\\
5	1.25793376063149\\
6	1.13259408972462\\
7	1.0698049875121\\
8	1.03068665481896\\
9	0.962039357769938\\
10	0.971727506631449\\
11	0.918412802784111\\
12	0.850425310426117\\
13	0.835799513185308\\
14	0.797238954566625\\
15	0.771436781097082\\
16	0.763393364915634\\
17	0.763145563063855\\
18	0.749483822449805\\
19	0.717147413484735\\
20	0.740830818507415\\
21	0.795675007022302\\
22	0.619051083339708\\
23	0.655282986781935\\
24	0.626201262828571\\
25	0.636596460049997\\
26	0.626834644396453\\
27	0.62631241699877\\
28	0.622837282890679\\
29	0.648724938140128\\
30	0.662692213712294\\
31	0.738821901340747\\
32	0.506077837320178\\
33	0.481235978511098\\
34	0.510938363328592\\
35	0.47335078977785\\
36	0.423401641814607\\
37	0.409350327336531\\
38	0.447219190505664\\
39	0.434340270681306\\
40	0.367235315752188\\
41	0.508527989227759\\
42	0.381381677193485\\
43	0.371062955680692\\
44	0.320350661486482\\
45	0.30122301928422\\
46	0.290443095184032\\
47	0.292547935847024\\
48	0.281333840028801\\
49	0.255582011324137\\
50	0.249835618879708\\
51	0.302321283910372\\
52	0.287951948622507\\
53	0.260749753665675\\
54	0.227909172485249\\
55	0.214036403396127\\
56	0.205003004671593\\
57	0.203106512729773\\
58	0.199405258716992\\
59	0.194980828544351\\
60	0.203391611162093\\
};
\addlegendentry{Algorithm~1}

\addplot [dashed,color=black, line width=1.0pt, mark=square, mark options={solid,thin, black}]
table[row sep=crcr]{%
1	2.23371168141479\\
2	1.56602148958878\\
3	1.44425891757145\\
4	1.33688718207588\\
5	1.29146207527951\\
6	1.18229405734122\\
7	1.09411403984678\\
8	1.02249891597808\\
9	1.00977382793559\\
10	0.975484384999005\\
11	0.884620264082975\\
12	0.860899282019633\\
13	0.815705031539916\\
14	0.776251281051934\\
15	0.762206220568393\\
16	0.746950287688975\\
17	0.736536877548282\\
18	0.708277897937623\\
19	0.692523536059337\\
20	0.687686344336248\\
21	0.692981918793985\\
22	0.630397529423752\\
23	0.649753138543298\\
24	0.635108695822418\\
25	0.676233895363993\\
26	0.597657026049508\\
27	0.60242312469563\\
28	0.588329562330546\\
29	0.556153431892646\\
30	0.552095240562925\\
31	0.565497427092058\\
32	0.514818013281758\\
33	0.43572587568718\\
34	0.426438038160923\\
35	0.416881883928778\\
36	0.440098621824861\\
37	0.401034981789137\\
38	0.374378633723769\\
39	0.345965505973984\\
40	0.34043312235982\\
41	0.403006797177156\\
42	0.34149058662127\\
43	0.293709934191031\\
44	0.281195395905638\\
45	0.241590300141487\\
46	0.240138879592101\\
47	0.257471676047311\\
48	0.215516511397964\\
49	0.222331943912578\\
50	0.20101939140825\\
51	0.254038319699437\\
52	0.238401933563133\\
53	0.194343842096429\\
54	0.177118291375535\\
55	0.164190690566362\\
56	0.161533874445664\\
57	0.153473960296989\\
58	0.143226468155771\\
59	0.149129316816343\\
60	0.16382514884684\\
};
\addlegendentry{Close-to-optimal}

\addplot [dotted,color=black, line width=1.0pt, mark=x, mark options={solid,thin, black}]
table[row sep=crcr]{%
1	2.25324232267214\\	
2	1.68463453660821\\	
3	1.52610782248699\\	
4	1.44084103780311\\	
5	1.35561831352349\\	
6	1.21233205474569\\	
7	1.17019495819473\\	
8	1.14167856312951\\	
9	1.14393414496331\\	
10	1.13865745741194\\	
11	1.2081817899147\\	
12	1.11114108974901\\	
13	1.19124392630542\\	
14	1.04305347061938\\	
15	1.11551883783146\\	
16	1.33544270124272\\	
17	1.27633632444559\\	
18	1.20699154740858\\	
19	1.25184496308238\\	
20	1.35105368459229\\	
21	1.26831433242596\\	
22	1.12054136543474\\	
23	1.09523821067464\\	
24	1.09911136332467\\	
25	1.35282905816552\\	
26	1.24182349316875\\	
27	1.46961453728425\\	
28	1.29293482251254\\	
29	1.29233262300861\\	
30	1.46861708058164\\	
31	1.08701178760607\\	
32	0.848167754302222\\	
33	0.882751672840585\\	
34	0.854196678371576\\	
35	0.830425344208666\\	
36	0.85780151615804\\	
37	0.889448762136807\\	
38	0.860649074819689\\	
39	0.847773727340308\\	
40	0.889777690433865\\	
41	0.969122843394747\\	
42	0.902158542553289\\	
43	1.02639918674946\\	
44	0.937189680094358\\	
45	0.923566613000551\\	
46	0.988284094363098\\	
47	1.04661787131053\\	
48	1.06624593376253\\	
49	0.998285820880193\\	
50	0.939226187286264\\	
51	0.96310594387819\\	
52	0.978360761087441\\	
53	0.940380798123436\\	
54	0.892544230832781\\	
55	1.03371181823058\\	
56	0.989898923108487\\	
57	0.934835118897229\\	
58	1.00822770455298\\	
59	1.00463769418272\\	
60	1.10742970275356\\	
};
\addlegendentry{Benchmark}

\addplot [dashdotdotted,color=black, line width=1.0pt]
table[row sep=crcr]{%
1	2.346172947\\
2	1.514006158\\
3	1.270416092\\
4	1.136136781\\
5	1.02962329\\
6	0.949405265\\
7	0.880000201\\
8	0.825641252\\
9	0.775108407\\
10	0.732911014\\
11	0.691830916\\
12	0.65760105\\
13	0.628139032\\
14	0.599246409\\
15	0.572571988\\
16	0.549266547\\
17	0.524959276\\
18	0.505531418\\
19	0.487350442\\
20	0.466033347\\
21	0.449172051\\
22	0.43383473\\
23	0.418800741\\
24	0.402336252\\
25	0.385773195\\
26	0.377128571\\
27	0.362236926\\
28	0.349182692\\
29	0.331243728\\
30	0.328231242\\
31	0.307848428\\
32	0.263258433\\
33	0.234962726\\
34	0.221564889\\
35	0.195920458\\
36	0.167823431\\
37	0.15283691\\
38	0.140936088\\
39	0.129376543\\
40	0.122845829\\
41	0.113562882\\
42	0.106834513\\
43	0.091876874\\
44	0.084305966\\
45	0.077778881\\
46	0.072517026\\
47	0.0681783\\
48	0.065015566\\
49	0.062264692\\
50	0.059808633\\
51	0.057732803\\
52	0.055878355\\
53	0.054205854\\
54	0.052727049\\
55	0.051301304\\
56	0.050268708\\
57	0.049405634\\
58	0.048628658\\
59	0.047938004\\
60	0.047295139\\
};
\addlegendentry{Perfect FL}
		
\end{axis}
\end{tikzpicture}%
	\caption{Training loss against the number of training epochs.}
	\label{fig:Training_Loss_vs_Epoch}
\end{figure}
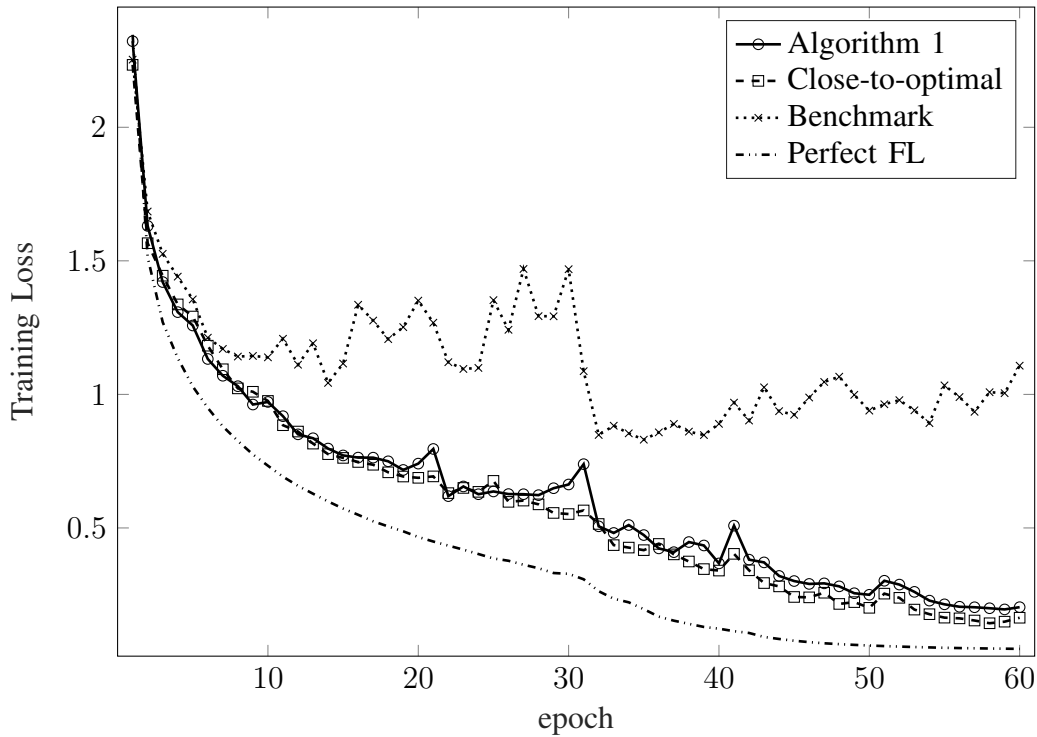

\begin{figure}
	\centering
%
%
\begin{tikzpicture}
	
	\begin{axis}[%
		width=4.8in,
		height=3.4in,
		at={(1.262in,0.697in)},
		scale only axis,
		xmin=0,
		xmax=61,
		xlabel style={font=\color{white!15!black}},
		xtick={10,20,30,40,50,60},
		xticklabels={{$10$},{$20$},{$30$},{$40$},{$50$},{$60$}},
		xlabel={epoch},
		ymin=0.02,
		ymax=2.45,
		ytick={0.5,1,1.5,2},
		yticklabels={{$0.5$},{$1$},{$1.5$},{$2$}},
		ylabel style={font=\color{white!15!black}},
		ylabel={Test Loss},
		axis background/.style={fill=none},
		legend style={legend cell align=left, align=left, draw=white!15!black}
		]

\addplot [color=black, line width=1.0pt, mark=o, mark options={thin, black}]
table[row sep=crcr]{%
1	2.32169438544318\\
2	1.63183533437342\\
3	1.42735402240313\\
4	1.32341711301145\\
5	1.28140361356816\\
6	1.16469363789979\\
7	1.11155740486013\\
8	1.07944004912671\\
9	1.02171503533181\\
10	1.04068482347863\\
11	0.989142293989253\\
12	0.932746541224343\\
13	0.927555446097438\\
14	0.902544358056739\\
15	0.891381850703467\\
16	0.894311478531523\\
17	0.907459040636198\\
18	0.908667229807122\\
19	0.888089723208045\\
20	0.926469421465036\\
21	0.947372084162194\\
22	0.793036458748866\\
23	0.845967637896642\\
24	0.835084396580474\\
25	0.864294554676777\\
26	0.872787279187315\\
27	0.888655170455923\\
28	0.901181328843395\\
29	0.945919326800767\\
30	0.974693721563121\\
31	0.988761020868311\\
32	0.77720356219903\\
33	0.780067144149466\\
34	0.841201342771807\\
35	0.82667773872905\\
36	0.808377571734001\\
37	0.826208387321713\\
38	0.889061889716824\\
39	0.897063777000282\\
40	0.840209152844461\\
41	0.909250920870379\\
42	0.795072676887126\\
43	0.815649640249319\\
44	0.798110091769219\\
45	0.814366431133667\\
46	0.829323508272146\\
47	0.852027251088056\\
48	0.856792278388609\\
49	0.8468649513004\\
50	0.852040478236397\\
51	0.848836470946415\\
52	0.857590153185175\\
53	0.845778828977871\\
54	0.82663629566139\\
55	0.825534493505684\\
56	0.820872284721165\\
57	0.828631106186322\\
58	0.829915889112652\\
59	0.828522131410107\\
60	0.846028371304939\\
};
\addlegendentry{Algorithm~1}

\addplot [dashed,color=black, line width=1.0pt, mark=square, mark options={solid,thin, black}]
table[row sep=crcr]{%
1	2.23200880580669\\
2	1.56604435762151\\
3	1.44857939305447\\
4	1.34856755830925\\
5	1.31107283989541\\
6	1.20866425615563\\
7	1.13020091705995\\
8	1.06643656658397\\
9	1.0604868043896\\
10	1.03508769892853\\
11	0.954334290674646\\
12	0.940535763809257\\
13	0.904524238468956\\
14	0.87610565225998\\
15	0.872057081264653\\
16	0.867856564315752\\
17	0.869839026687219\\
18	0.851773792603236\\
19	0.848431142457395\\
20	0.855847874941896\\
21	0.864319033754005\\
22	0.812198683387971\\
23	0.844152594439471\\
24	0.839932288392949\\
25	0.895834284936704\\
26	0.826816343275948\\
27	0.846844553535761\\
28	0.849914336996637\\
29	0.831678709215659\\
30	0.849239267556549\\
31	0.83827976543396\\
32	0.810660787455612\\
33	0.763786314277546\\
34	0.778536267096479\\
35	0.787971735764727\\
36	0.844156829659588\\
37	0.816120361767943\\
38	0.819091218344424\\
39	0.806486913481093\\
40	0.813606461665134\\
41	0.844679766582946\\
42	0.830027816304044\\
43	0.796422935490009\\
44	0.804591399030513\\
45	0.771643039137641\\
46	0.791068592853397\\
47	0.822760192694222\\
48	0.792599898464127\\
49	0.809502330551808\\
50	0.795060544179745\\
51	0.848312225076807\\
52	0.853697171671653\\
53	0.81340017368059\\
54	0.810919750195477\\
55	0.804167708756882\\
56	0.806315997477034\\
57	0.807026388434572\\
58	0.796332661559729\\
59	0.812610604658215\\
60	0.835167992551087\\
};
\addlegendentry{Close-to-optimal}

\addplot [dotted,color=black, line width=1.0pt, mark=x, mark options={solid,thin, black}]
table[row sep=crcr]{%
1	2.25231563083464\\
2	1.68293373672892\\
3	1.52912759604119\\
4	1.44861853851212\\
5	1.37022091454025\\
6	1.23705753876175\\
7	1.2021139800215\\
8	1.18158258659948\\
9	1.18855125983175\\
10	1.19348129433193\\
11	1.25132589858951\\
12	1.16575951758239\\
13	1.25113561865714\\
14	1.10964354113054\\
15	1.19151576746592\\
16	1.42196732352195\\
17	1.36724259827456\\
18	1.3071574876819\\
19	1.35999682902438\\
20	1.47018486958036\\
21	1.35346557977472\\
22	1.21688433412126\\
23	1.19984352306651\\
24	1.21251454693463\\
25	1.47707853181233\\
26	1.37815608017349\\
27	1.6131517292986\\
28	1.45374396946356\\
29	1.46348042636135\\
30	1.63196345546185\\
31	1.19173995786377\\
32	0.977761725392818\\
33	1.02703102276914\\
34	1.02010506332177\\
35	1.0145984271201\\
36	1.05599493425072\\
37	1.10038261899182\\
38	1.08612675315203\\
39	1.08380082159372\\
40	1.14015929641088\\
41	1.13806861157028\\
42	1.07356993375766\\
43	1.20467171399228\\
44	1.12447524352519\\
45	1.11619467633761\\
46	1.18971761247258\\
47	1.24474045863006\\
48	1.27536196978063\\
49	1.21869592013193\\
50	1.15478346201608\\
51	1.14060976023955\\
52	1.17111483521858\\
53	1.13768749730312\\
54	1.10640411473363\\
55	1.25939419651541\\
56	1.21823734613145\\
57	1.16445624342235\\
58	1.24933180629989\\
59	1.24959860514659\\
60	1.35507237303637\\
};
\addlegendentry{Benchmark}

\addplot [dashdotdotted,color=black, line width=1.0pt]
table[row sep=crcr]{%
1	2.34429402490028\\
2	1.51759445077653\\
3	1.28127905090007\\
4	1.15490306443233\\
5	1.05786668352908\\
6	0.987319113383834\\
7	0.928235604949921\\
8	0.883608958432429\\
9	0.843722519193658\\
10	0.812128130119625\\
11	0.782310208424533\\
12	0.760217778617851\\
13	0.743251841136073\\
14	0.727556170350749\\
15	0.714245304331291\\
16	0.705182474987398\\
17	0.695625988769678\\
18	0.692070666961015\\
19	0.689693812569\\
20	0.683678724669474\\
21	0.685488668302123\\
22	0.687403815149439\\
23	0.689018945195723\\
24	0.689150830796346\\
25	0.690727702003057\\
26	0.69968126334738\\
27	0.703927341761007\\
28	0.705559272527174\\
29	0.703878809744716\\
30	0.718346326207218\\
31	0.712119303235887\\
32	0.683529466360347\\
33	0.67600357763139\\
34	0.685378710930688\\
35	0.6790569537848\\
36	0.672653176382768\\
37	0.682141663807171\\
38	0.691724105018087\\
39	0.700983479096079\\
40	0.713888210423671\\
41	0.72203674792926\\
42	0.731049285044094\\
43	0.729853974872199\\
44	0.737141803892594\\
45	0.742601984936773\\
46	0.747281824194328\\
47	0.751339405253641\\
48	0.756114987432804\\
49	0.760630120763855\\
50	0.764764991983174\\
51	0.768659803656958\\
52	0.772445611363623\\
53	0.776089368816197\\
54	0.779341398262135\\
55	0.781745118628629\\
56	0.784008701755442\\
57	0.786032603242051\\
58	0.787936343694055\\
59	0.789678981245163\\
60	0.791341617035793\\
};
\addlegendentry{Perfect FL}
		
\end{axis}
\end{tikzpicture}%
	\caption{Test loss against the number of training epochs.}
	\label{fig:Loss_vs_Epoch}
\end{figure}
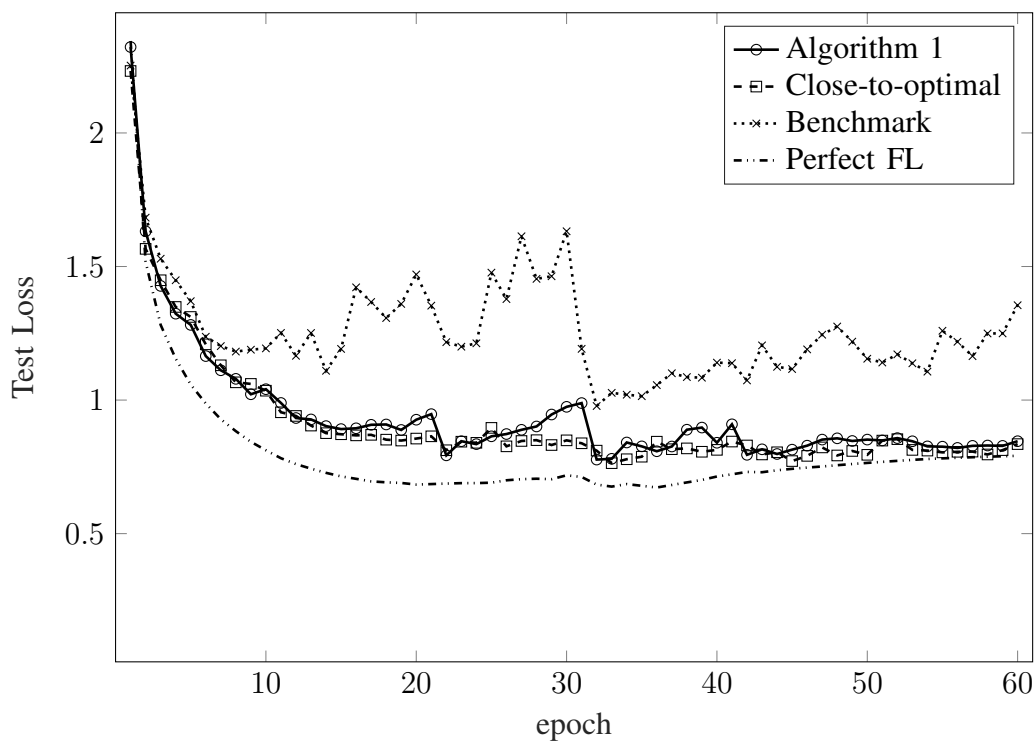

As Fig.~\ref{fig:Loss_vs_Epoch} shows, the test losses incurred by both the algorithms closely track the loss achieved by perfect \ac{FL} over a noise-less network while the benchmark algorithm performs considerably degraded. For this figure, the tolerable computation error is set such that $\log \gamma = 15$ dB. Compared with the close-to-optimal policy, the minor degradation observed in the performance of the proposed scheme is compensated for in terms of complexity. 

\begin{figure}
	\centering
%
%
\begin{tikzpicture}
	
	\begin{axis}[%
		width=4.8in,
		height=3.4in,
		at={(1.262in,0.697in)},
		scale only axis,
		xmin=0,
		xmax=61,
		xlabel style={font=\color{white!15!black}},
		xtick={10,20,30,40,50,60},
		xticklabels={{$10$},{$20$},{$30$},{$40$},{$50$},{$60$}},
		xlabel={epoch},
		ymin=10,
		ymax=89,
		ytick={20,40,60,80,100},
		yticklabels={{$20$},{$40$},{$60$},{$80$},{$100$}},
		ylabel style={font=\color{white!15!black}},
		ylabel={Accuracy in \%},
		axis background/.style={fill=none},
		legend style={at={(.98,.28)},legend cell align=left, align=left, draw=white!15!black}
		]

\addplot [color=black, line width=1.0pt, mark=o, mark options={thin, black}]
table[row sep=crcr]{%
1	13.2753333333333\\
2	39.6966666666667\\
3	48.1173333333333\\
4	52.252\\
5	54.7353333333333\\
6	58.4013333333333\\
7	60.6106666666667\\
8	61.9406666666667\\
9	64.0346666666667\\
10	63.9266666666667\\
11	65.1993333333333\\
12	67.1393333333333\\
13	67.54\\
14	68.658\\
15	69.3966666666667\\
16	69.6213333333333\\
17	69.5653333333333\\
18	69.8573333333333\\
19	70.5633333333333\\
20	69.9\\
21	68.646\\
22	73.4793333333333\\
23	72.1406666666667\\
24	72.9586666666667\\
25	72.6093333333333\\
26	72.7446666666667\\
27	72.866\\
28	72.6466666666667\\
29	72.184\\
30	72.2126666666667\\
31	70.3553333333333\\
32	74.9093333333333\\
33	75.316\\
34	74.526\\
35	75.0766666666667\\
36	76.094\\
37	76.302\\
38	75.1606666666667\\
39	75.5213333333333\\
40	76.9093333333333\\
41	73.634\\
42	76.228\\
43	76.124\\
44	77.1193333333333\\
45	77.394\\
46	77.6013333333333\\
47	77.398\\
48	77.7173333333333\\
49	78.1706666666667\\
50	78.2253333333333\\
51	77.248\\
52	77.1493333333333\\
53	77.66\\
54	78.3073333333333\\
55	78.4806666666667\\
56	78.694\\
57	78.6713333333333\\
58	78.726\\
59	78.7973333333333\\
60	78.6206666666667\\
};
\addlegendentry{Algorithm~1}

\addplot [dashed,color=black, line width=1.0pt, mark=square, mark options={solid,thin, black}]
table[row sep=crcr]{%
1	15.9546666666667\\
2	42.3593333333333\\
3	47.7473333333333\\
4	51.8246666666667\\
5	53.6753333333333\\
6	57.4886666666667\\
7	59.768\\
8	61.9833333333333\\
9	62.552\\
10	63.8806666666667\\
11	66.586\\
12	67.4526666666667\\
13	68.588\\
14	69.5026666666667\\
15	69.9066666666667\\
16	70.158\\
17	70.466\\
18	71.158\\
19	71.3346666666667\\
20	71.502\\
21	71.2393333333333\\
22	72.7566666666667\\
23	72.256\\
24	72.2826666666667\\
25	71.396\\
26	73.1973333333333\\
27	72.8793333333333\\
28	73.328\\
29	73.8453333333333\\
30	73.906\\
31	73.1706666666667\\
32	74.1033333333333\\
33	76.006\\
34	75.9866666666667\\
35	75.9206666666667\\
36	75.2393333333333\\
37	75.978\\
38	76.448\\
39	76.7706666666667\\
40	76.616\\
41	75.9833333333334\\
42	76.4933333333333\\
43	77.292\\
44	77.3446666666667\\
45	78.2906666666667\\
46	78.1206666666667\\
47	77.768\\
48	78.3826666666667\\
49	78.134\\
50	78.5706666666667\\
51	77.442\\
52	77.5413333333333\\
53	78.4166666666667\\
54	78.6606666666667\\
55	78.9166666666667\\
56	78.9686666666667\\
57	78.978\\
58	79.312\\
59	79.144\\
60	78.824\\
};
\addlegendentry{Close-to-optimal}

\addplot [dotted,color=black, line width=1.0pt, mark=x, mark options={solid,thin, black}]
table[row sep=crcr]{%
1	16.0126666666667\\
2	38.322\\
3	44.51\\
4	48.7946666666667\\
5	51.146\\
6	55.896\\
7	57.2686666666667\\
8	58.5193333333333\\
9	58.8833333333333\\
10	59.226\\
11	56.9886666666667\\
12	60.3626666666667\\
13	60.426\\
14	61.9466666666667\\
15	61.8046666666667\\
16	59.192\\
17	59.556\\
18	62.8326666666667\\
19	60.57\\
20	60.1053333333333\\
21	58.61\\
22	62.8613333333334\\
23	63.3526666666667\\
24	63.2173333333333\\
25	61.6273333333333\\
26	61.72\\
27	61.478\\
28	61.7966666666667\\
29	61.9533333333333\\
30	62.4213333333333\\
31	60.4806666666667\\
32	66.61\\
33	65.9713333333334\\
34	66.95\\
35	67.1726666666667\\
36	67.174\\
37	66.1193333333333\\
38	66.2706666666667\\
39	66.768\\
40	66.2726666666667\\
41	63.5966666666667\\
42	65.0966666666667\\
43	63.8633333333333\\
44	65.0733333333333\\
45	64.9746666666667\\
46	64.1886666666667\\
47	63.9253333333333\\
48	63.6233333333333\\
49	64.354\\
50	65.1126666666667\\
51	64.0253333333334\\
52	64.2453333333333\\
53	65.088\\
54	66.0133333333333\\
55	63.93\\
56	64.176\\
57	64.924\\
58	64.1586666666667\\
59	63.9646666666667\\
60	62.988\\
};
\addlegendentry{Benchmark}

\addplot [dashdotdotted,color=black, line width=1.0pt]
table[row sep=crcr]{%
1	14.895\\
2	44.334\\
3	53.813\\
4	58.58\\
5	62.072\\
6	64.56\\
7	66.808\\
8	68.48\\
9	70.015\\
10	71.237\\
11	72.442\\
12	73.283\\
13	74.007\\
14	74.726\\
15	75.346\\
16	75.722\\
17	76.238\\
18	76.48\\
19	76.745\\
20	77.147\\
21	77.264\\
22	77.513\\
23	77.595\\
24	77.89\\
25	78.001\\
26	78.09\\
27	78.209\\
28	78.433\\
29	78.553\\
30	78.482\\
31	78.764\\
32	79.379\\
33	79.609\\
34	79.645\\
35	79.936\\
36	80.219\\
37	80.275\\
38	80.367\\
39	80.372\\
40	80.404\\
41	80.399\\
42	80.364\\
43	80.597\\
44	80.642\\
45	80.64\\
46	80.649\\
47	80.664\\
48	80.653\\
49	80.662\\
50	80.682\\
51	80.712\\
52	80.694\\
53	80.703\\
54	80.735\\
55	80.73\\
56	80.735\\
57	80.74\\
58	80.755\\
59	80.758\\
60	80.76\\
};
\addlegendentry{Perfect FL}
		
	\end{axis}
\end{tikzpicture}%
	\caption{Accuracy of classification over the test dataset against the number of training epochs.}
	\label{fig:Accuracy_vs_Epoch}
\end{figure}
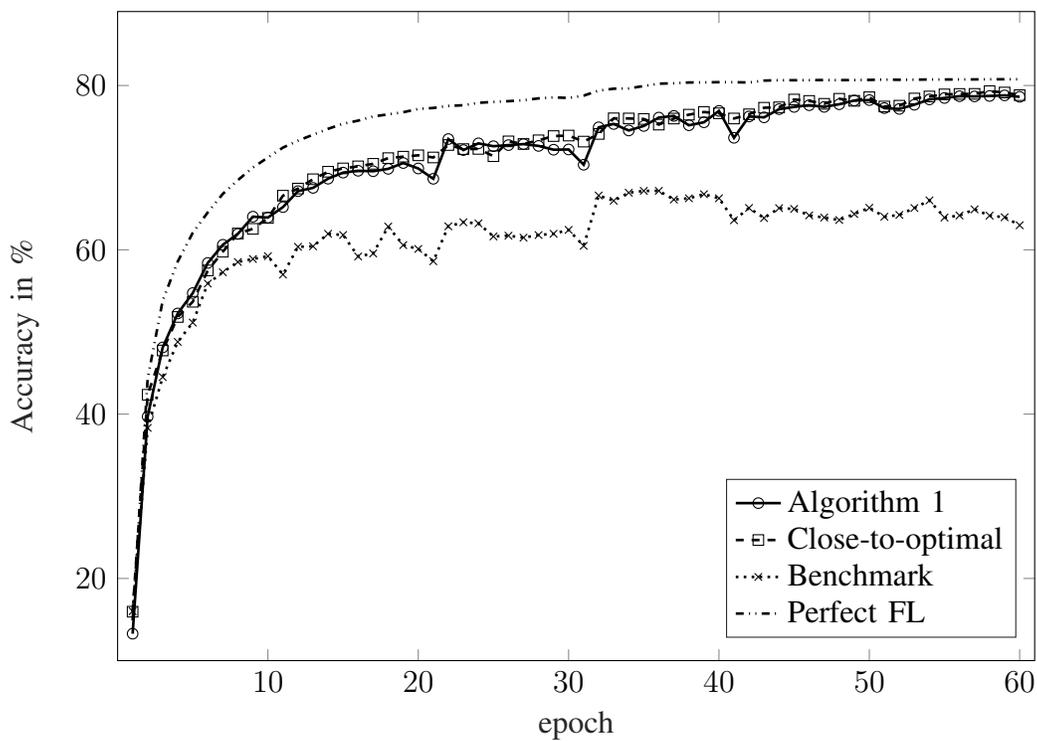

Fig.~\ref{fig:Accuracy_vs_Epoch} shows the accuracy of the trained machines in classifying the test dataset. As expected from the behavior of the test loss, the matching pursuit scheme and the close-to-optimal policy closely track the upper-bound achieved by perfect \ac{FL} while the benchmark algorithm significantly performs degraded. This observation can be illustrated as follows: While the close-to-optimal approach schedules almost optimally with high computational complexity, the matching pursuit strategy reduces the complexity significantly at the expense of a sub-optimal performance. This degradation is however negligible, since the matching pursuit technique drops only few devices with minimal effect in the overall learning performance. The benchmark algorithm however fails to schedule dominant devices, and hence performs considerably degraded.

\begin{figure}
	\centering
%
%
\begin{tikzpicture}
	
	\begin{axis}[%
		width=4.8in,
		height=3.4in,
		at={(1.262in,0.697in)},
		scale only axis,
		xmin=-12,
		xmax=52,
		xlabel style={font=\color{white!15!black}},
		xtick={-10,0,10,20,30,40,50},
		xticklabels={{$-10$},{$0$},{$10$},{$20$},{$30$},{$40$},{$50$}},
		xlabel={$\log \gamma$ in [dB]},
		ymin=-.5,
		ymax=21,
		ytick={0,5,10,15,20},
		yticklabels={{$0$},{$5$},{$10$},{$15$},{$20$}},
		ylabel style={font=\color{white!15!black}},
		ylabel={$\Ex{\abs{\setS}}{ }$},
		axis background/.style={fill=none},
		legend style={at={(.98,.22)},legend cell align=left, align=left, draw=white!15!black}
		]

\addplot [color=black, line width=1.0pt, mark=o, mark options={thin, black}]
table[row sep=crcr]{%
-10	1.355\\
-5	2.275\\
0	3.84\\
5	6.385\\
10	9.55\\
15	12.865\\
20	15.7\\
25	17.795\\
30	19.125\\
35	19.71\\
40	19.875\\
45	19.935\\
50	19.985\\
};
\addlegendentry{Algorithm~1}

\addplot [dashed,color=black, line width=1.0pt, mark=square, mark options={solid,thin, black}]
table[row sep=crcr]{%
-10	1.38\\
	-5	3.39\\
	0	6.915\\
	5	11.195\\
	10	15.17\\
	15	18.48\\
	20	19.65\\
	25	19.875\\
	30	19.865\\
	35	19.94\\
	40	19.965\\
	45	19.97\\
	50	19.965\\
};
\addlegendentry{Close-to-optimal}

\addplot [dotted,color=black, line width=1.0pt, mark=x, mark options={solid,thin, black}]
table[row sep=crcr]{%
-10	0.68\\
-5	1.35\\
0	1.74\\
5	2.895\\
10	4.895\\
15	6.82\\
20	11.79\\
25	16.445\\
30	18.445\\
35	19.65\\
40	19.805\\
45	19.95\\
50	19.97\\
};
\addlegendentry{Benchmark}

%
		
	\end{axis}
\end{tikzpicture}%
	\caption{Average $\abs{\setS}$ versus computation error tolerance for the simulated network.}
	\label{fig:Nr_vs_gamma_Channel}
\end{figure}
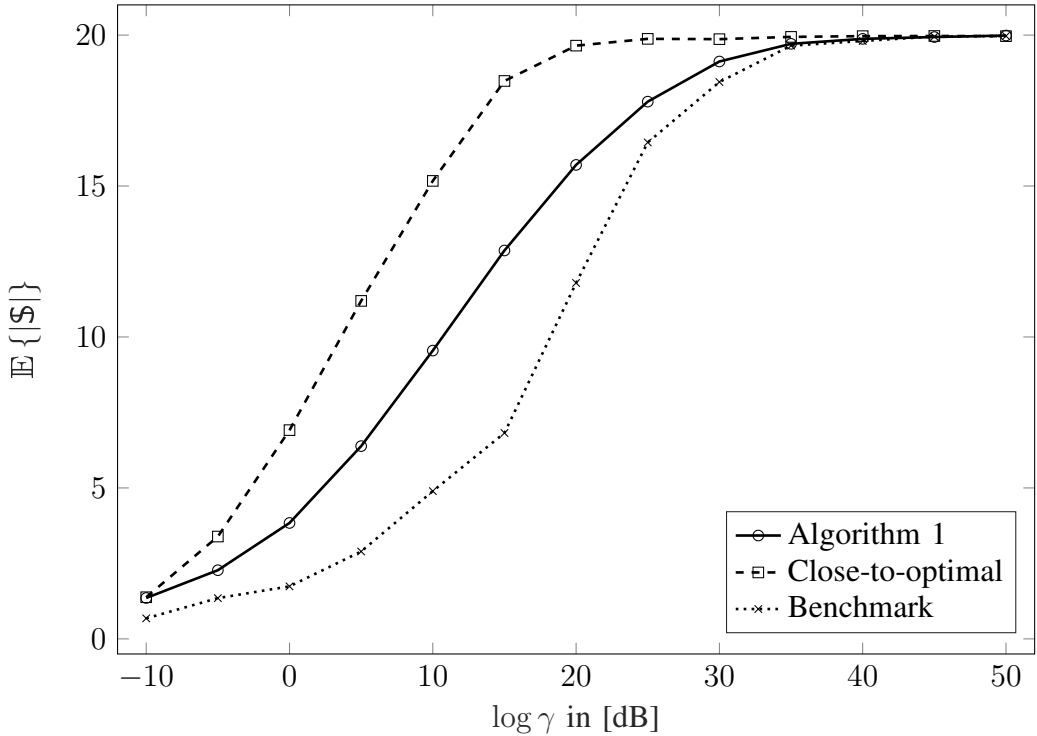

To specify the scheduling performance, we repeat the experiment in Fig.~\ref{fig:Nr_vs_gamma} for the simulated channels\footnote{Remember that in Fig.~\ref{fig:Nr_vs_gamma} the channels are generated \ac{iid} Gaussian.} and compare the proposed scheduling algorithm with the close-to-optimal and the benchmark in Fig.~\ref{fig:Nr_vs_gamma_Channel}. As the figure shows, the close-to-optimal algorithm outperforms Algorithm~\ref{alg:MPS} while both outperform the benchmark. As mentioned, this degradation in performance results in a significant complexity reduction which is demonstrated in Fig.~\ref{fig:Runtime_Channel}. In this figure, we plot the runtime of each algorithm against the number of devices $K$ and the receiver dimension $N$ while considering low error tolerance, i.e., $\log\gamma = 0$ dB.  To be able to compare the figures, we represent the vertical axis in the logarithmic scale. As the figure shows, using the proposed scheme, the complexity reduces significantly compared with both the close-to-optimal and benchmark while performing in between the two schemes. This shows a clear enhancement with respect to the complexity-performance trade-off.

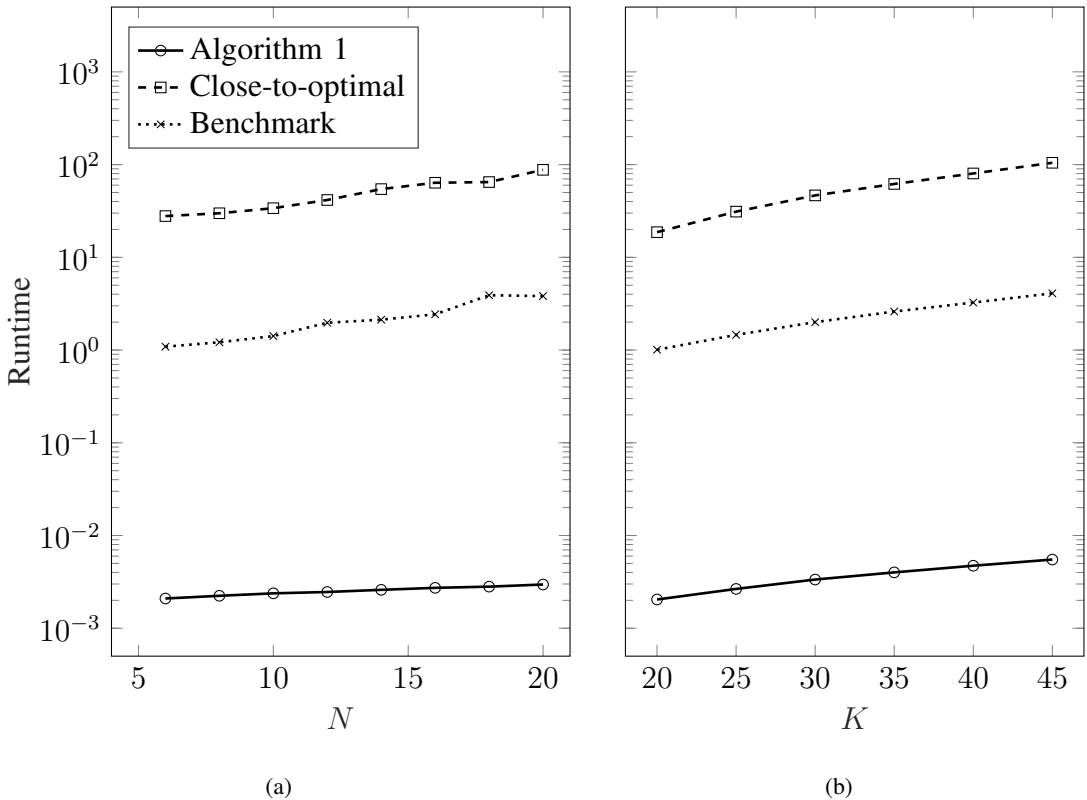
\begin{figure}
	\begin{center}
		\subfigure[ ]{
%
%
\begin{tikzpicture}
	
	\begin{axis}[%
		width=2.4in,
		height=3.4in,
		at={(1.262in,0.697in)},
		scale only axis,
		xmin=4,
		xmax=21,
		xlabel style={font=\color{white!15!black}},
		xtick={0,5,10,15,20},
		xticklabels={{$0$},{$5$},{$10$},{$15$},{$20$}},
		xlabel={$N$},
		ymin=0.0005,
		ymax=5000,
		ymode=log,
		ytick={0.001,.01,.1,1,10,100,1000},
		yticklabels={{$10^{-3}$},{$10^{-2}$},{$10^{-1}$},{$10^{0}$},{$10^{1}$},{$10^{2}$},{$10^{3}$}},
		ylabel style={font=\color{white!15!black}},
		ylabel={Runtime},
		axis background/.style={fill=none},
		legend style={at={(.67,.97)},legend cell align=left, align=left, draw=white!15!black}
		]

\addplot [color=black, line width=1.0pt, mark=o, mark options={thin, black}]
table[row sep=crcr]{%
6	0.002093399\\
8	0.002237883\\
10	0.002378967\\
12	0.002459953\\
14	0.002595873\\
16	0.002726722\\
18	0.002809277\\
20	0.00296237\\
};
\addlegendentry{Algorithm~1}

\addplot [dashed,color=black, line width=1.0pt, mark=square, mark options={solid,thin, black}]
table[row sep=crcr]{%
6	27.85725\\
8	29.87768\\
10	33.88828\\
12	41.53339\\
14	54.42133\\
16	63.71647\\
18	64.81365\\
20	87.86103\\
};
\addlegendentry{Close-to-optimal}

\addplot [dotted,color=black, line width=1.0pt, mark=x, mark options={solid,thin, black}]
table[row sep=crcr]{%
6	1.087493\\
8	1.213594\\
10	1.416067\\
12	1.969885\\
14	2.12686\\
16	2.430087\\
18	3.898592\\
20	3.823677\\
};
\addlegendentry{Benchmark}

%
		
	\end{axis}
\end{tikzpicture}%
		}
		\subfigure[ ]{
%
%
\begin{tikzpicture}
	
	\begin{axis}[%
		width=2.4in,
		height=3.4in,
		at={(1.262in,0.697in)},
		scale only axis,
		xmin=18,
		xmax=47,
		xlabel style={font=\color{white!15!black}},
		xtick={20,25,30,35,40,45},
		xticklabels={{$20$},{$25$},{$30$},{$35$},{$40$},{$45$}},
		xlabel={$K$},
		ymin=0.0005,
		ymax=5000,
		ymode=log,
		ytick={0.001,.01,.1,1,10,100,1000},
		yticklabels={{ },{ },{ },{ },{ },{ },{ }},
		ylabel style={font=\color{white!15!black}},
		ylabel={ },
		axis background/.style={fill=none},
		legend style={at={(.98,.22)},legend cell align=left, align=left, draw=white!15!black}
		]

\addplot [color=black, line width=1.0pt, mark=o, mark options={thin, black},forget plot]
table[row sep=crcr]{%
20	0.00204128000\\
25	0.00265485000\\
30	0.00335073000\\
35	0.00400329700\\
40	0.00472350700\\
45	0.00550467800\\
};

\addplot [dashed,color=black, line width=1.0pt, mark=square, mark options={solid,thin, black},forget plot]
table[row sep=crcr]{%
20	18.6729\\
25	31.11323\\
30	46.5429\\
35	61.8586\\
40	80.25945\\
45	104.6096\\
};

\addplot [dotted,color=black, line width=1.0pt, mark=x, mark options={solid,thin, black},forget plot]
table[row sep=crcr]{%
20	1.00767\\
25	1.458137\\
30	2.0001\\
35	2.6035\\
40	3.25318\\
45	4.08893\\
};

%
		
	\end{axis}
\end{tikzpicture}%
		}
	\end{center}
	\caption{Runtime in seconds against network dimensions.}
	\label{fig:Runtime_Channel}
\end{figure}

To evaluate the impact of over-the-air computation on the performance of \ac{FL}, we define a new efficiency metric. Namely, we define the \textit{over-the-air efficiency} as 
\begin{align}
	\zeta_{\mathrm{ota} } = \frac{\text{Accuracy}_{\rm ota}}{\text{Accuracy}_{\rm fl}}
\end{align}
where $\text{Accuracy}_{\rm ota}$ and $\text{Accuracy}_{\rm fl}$ represent the classification accuracy achieved by \textit{over-the-air computation} and \textit{perfect \ac{FL}}, respectively. In general, $ 0 \leq	\zeta_{\mathrm{ota} } \leq 1$ and characterizes the loss imposed on federated averaging, due to the imperfection of \ac{otaFL}. Fig.~\ref{fig:OTALoss_vs_Gamma} shows the over-the-air efficiency as a function of the error tolerance $\gamma$ for $K=20$ devices and $N=6$ antennas when we train the model with $60$ epochs, i.e., after six communication rounds, using the proposed scheduling scheme. Interestingly, the figure shows a clear trend: there exists a tolerance level at which the algorithm results in best efficiency. This behavior can be clarified in light of the contradiction between the learning performance enhancement and the degradation caused by imperfect aggregation: For small choices of $\gamma$, few devices are selected. In this case, the gain achieved by restricting the aggregation error is inferior to the loss imposed on the inference performance, due to training over a small dataset. As a result, the over-the-air efficiency improves by loosening the constraint on the aggregation error, i.e., allowing more devices to participate. This improvement continues until an optimal point is reached, where increasing the level of tolerable aggregation error\footnote{Due to imperfect aggregation.} severely degrades the overall learning performance, such that the enhancement achieved by learning over a larger dataset is not dominant anymore. By going above this optimal point, the over-the-air efficiency only decreases.   

As the figure shows, by properly choosing the tolerance level, we can achieve up to $98\%$ of over-the-air efficiency. This follows from the fact that some of devices have small local datasets, and their contribution in the learning task is rather minor. It is worth mentioning that $98\%$ of the over-the-air efficiency in this case is achieved by Algorithm~\ref{alg:MPS} which significantly reduces the scheduling complexity.

\begin{figure}
	\centering
%
%
\begin{tikzpicture}
	
	\begin{axis}[%
		width=4.8in,
		height=3.4in,
		at={(1.262in,0.697in)},
		scale only axis,
		xmin=-2,
		xmax=42,
		xlabel style={font=\color{white!15!black}},
		xtick={0,10,20,30,40},
		xticklabels={{$0$},{$10$},{$20$},{$30$},{$40$}},
		xlabel={$\log \gamma$ in [dB]},
		ymin=.7,
		ymax=1.05,
		ytick={.4,.60,.80, 0.979222387320456,1.00,1.2},
		yticklabels={{$0.4$},{$0.6$},{$0.8$},{$0.98$},{$1$},{$1.2$}},
		ylabel style={font=\color{white!15!black}},
		ylabel={$\zeta_{\mathrm{ota} }$},
		axis background/.style={fill=none},
		legend style={at={(.98,.22)},legend cell align=left, align=left, draw=white!15!black}
		]

\addplot [color=black, line width=1.0pt, mark=o, mark options={thin, black},forget plot]
table[row sep=crcr]{%
0	0.853120356612184\\
5	0.899827\\
10	0.955721\\
15	0.970604259534423\\
20	0.979222387320456\\
25	0.975482912332838\\
30	0.95269935611689\\
35	0.904011887072808\\
40	0.811267954432888\\
};

\addplot [dotted,color=black, line width=.5pt,forget plot]
table[row sep=crcr]{%
	20	-1\\
	20	0.979222387320456\\
};

\addplot [dotted,color=black, line width=.5pt,forget plot]
table[row sep=crcr]{%
	-4	0.979222387320456\\
	20	0.979222387320456\\
};

%
%
%
		
	\end{axis}
\end{tikzpicture}%
	\caption{Over-the-air efficiency against the computation error tolerance.}
	\label{fig:OTALoss_vs_Gamma}
\end{figure}
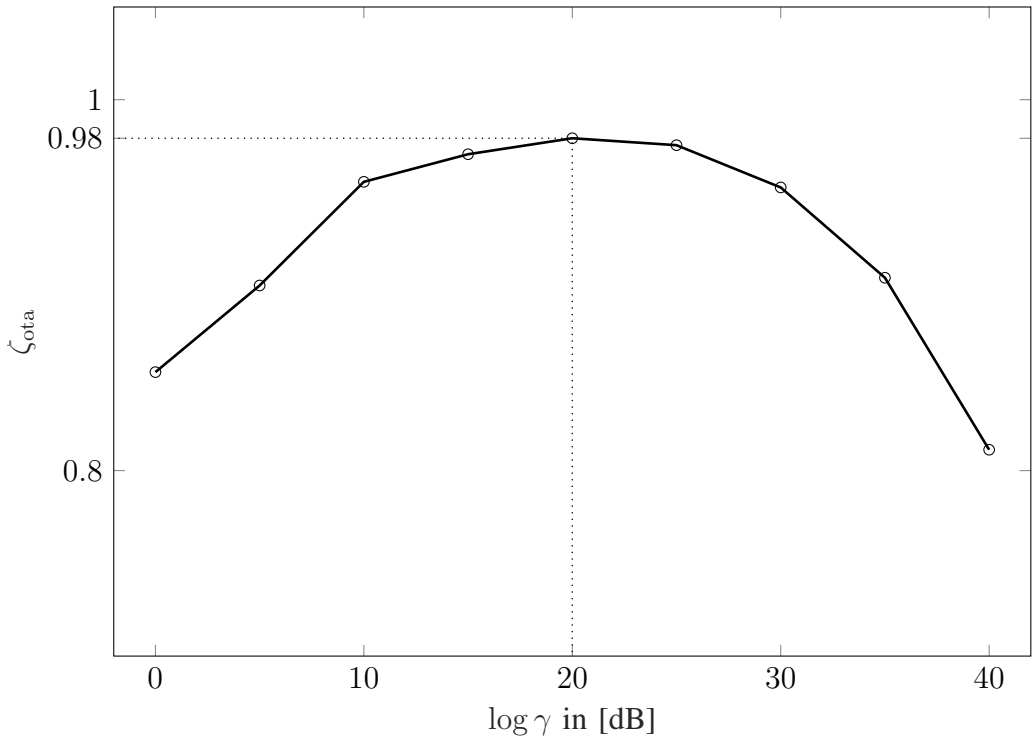


\section{Conclusions}
\label{sec:conc}
Considering the limited computational capacity of communication devices in wireless networks, the proposed scheme should be an efficient approach to perform \ac{otaFL} in various recently-proposed network architectures. An example is the concept of  \textit{ultra-dense networks} in which a very large number of devices with limited power and storage are to be connected with a single access point \cite{hwang2013holistic,bhushan2014network,andrews2016we,kamel2016ultra,adedoyin2020combination,kim2021energy}. The limited storage at edge devices in such networks leads to poor learning performance, when it is performed individually. On the other hand, the large number of devices can lead to high aggregation cost at the access point, if it applies an \ac{FL} algorithm on the complete dataset via a separately-designed aggregation strategy. In such scenarios, the matching-pursuit-based scheduling scheme allows a fair trade-off between the learning performance and the aggregation cost with a manageable complexity.

The numerical results imply that the over-the-air learning-aggregation trade-off shows a trend: With a predefined scheme for \ac{otaFL}, there exists an optimal operating point. By deviating from this point via hardening or loosening the constraint on the aggregation cost, the overall learning performance is degraded. This behavior further confirms our initial statement: Device scheduling is in general \textit{sub-optimal}; however,~with~a predefined setting and a restricted complexity budget, it can efficiently address the challenges in \ac{otaFL}. Although design of the optimal joint \ac{otaFL} strategy seems to be analytically intractable, our intuition based on properties of multiple access channels suggests that the performance of optimal \ac{aircomp}-based \ac{otaFL} can be characterized in the asymptotic regime. This is an interesting direction for future work.

The presented work can be extended in various directions. One natural direction is to extend the proposed scheme to networks with parametrized channels, e.g., \ac{irs}-aided networks \cite{wu2019towards}. This direction is briefly discussed in Appendix~\ref{app:A}. Investigating the robustness of the proposed scheme under various imperfections, e.g., imperfect \ac{csi} acquisition, is another direction for further study. 

\appendices
\section{Background on Over-the-Air Federated Learning}
\label{App:1}
The concept of \ac{otaFL} via \ac{aircomp} can be clearly illustrated through an example. In this appendix, we give a simple comprehensive example of \ac{FL}. Using this example, we then clarify the concept of \ac{aircomp}-based \ac{otaFL}.

\subsection{An Illustrative Example of Federated Learning}
Consider a basic training task in a machine learning problem: We intend to fit a given dataset onto a \textit{linear} model; namely, a set of data pairs $\setD = \set{\brc{\beta_i, \baa_i} : \forall i\in\dbc{I} }$ is available with $\beta_i$ being a real scalar and $\baa_i\in\setR^D$. We intend to learn the following linear model:
\begin{align}
	\beta = \btheta^\trp \baa \label{eq:model}
\end{align}
from the dataset, i.e., we aim to find $\btheta\in\setR^D$ such that \eqref{eq:model} describes the dataset with minimum error.

A classical approach to address this learning task is to find the model parameters, i.e., the entries of $\btheta$, via the method of \ac{ls}. This is a trivial task when the dataset is centralized at the \ac{ps}. In this case, the \ac{ps} sets the \ac{rss} as a loss function, i.e., 
\begin{align}
	F\brc{\btheta} = \sum_{i=1}^I \abs{\beta_i - \btheta^\trp \baa_i}^2,
\end{align}
and finds the model parameters by solving the following optimization problem
\begin{align}
	\min_{\btheta\in\setR^D} F\brc{\btheta}. \label{eq:Center}
\end{align}

The above task is trivial due to the centralized nature of the setting; however, in decentralized settings it becomes more challenging: Assume that the dataset $\setD$ is partitioned into sub-sets $\setD_k$ for $k\in\dbc{K}$ and shared among $K$ devices. In this case, finding the model parameters through a joint optimization, i.e., like the one given in \eqref{eq:Center}, requires that the devices send their \textit{local} datasets, i.e., $\setD_k$, into a central node. This is however not feasible for a wide variety of applications. The key issues that prevent using such a centralized setting are \textit{privacy} and \textit{communication}. In fact, by sending the local datasets to a central dataset, we increase the risk of local information being revealed to untrustworthy parties. Moreover, in many applications, the aggregate of the local datasets is \textit{large}  and sharing the complete local datasets leads to network overload.

To deal with decentralized settings, we mainly have two options:
\begin{itemize}
	\item Either let each device learn its own model parameters independently, i.e., device $k$ solves \eqref{eq:Center} for an \ac{rss} term determined by the data points in $\setD_k$.
	\item Or, use an \ac{FL} algorithm to learn a common model from the \textit{local} model parameters.
\end{itemize}
The \ac{FL} approach is of great interest in various applications, as it enables a trade-off between the intractability of centralized learning and the inefficiency of independent distributed learning. 

Let us now get back to the considered example assuming that the datasets is distributed among $K$ devices. We consider the \textit{federated averaging} algorithm \cite{mcmahan2017communication}: Each device finds its own model parameter $\btheta_k$ by solving 
\begin{align}
	\min_{\btheta\in\setR^D} F_k\brc{\btheta},
\end{align}
where $F_k\brc{\btheta}$ denotes the locally-calculated \ac{rss} of device $k$, i.e., 
\begin{align}
	F_k\brc{\btheta} = \sum_{\brc{\beta,\baa} \in \setD_k } \abs{\beta - \btheta^\trp \baa}^2.
\end{align}
It then sends $\btheta_k$ to the \ac{ps} in the network. After collecting all the local model parameters, the \ac{ps} sets the \textit{global} model parameter to  be
\begin{align}
	\bar{\btheta} = \sum_{k=1}^{K} \phi_k \btheta_k, \label{eq:global_model}
\end{align}
and shares it with the devices. Here, $0 \leq \phi_k\leq 1$ for $k\in\dbc{K}$ are weighting coefficients chosen with respect to the reliability of the local datasets\footnote{For instance, $\phi_k$ is chosen proportional to the size of the local dataset at device $k$.}. In this algorithm, one can observe $\bar{\btheta}$ as a simple approximation of the optimal model parameter derived intuitively from this observation that 
\begin{align}
	F\brc{\btheta}  = \sum_{k=1}^K F_k\brc{\btheta}.
\end{align}

In general, one can think of a more generic computational setting in which each device uploads a \textit{function} of its mode parameters and the \ac{ps} applies \ac{FL} by determining a \textit{function} of the \textit{aggregated} model parameters. In this respect, the federated averaging algorithm is seen as a special case in which all functions are set to be  \textit{linear}. 

\subsection{OTA-FL via AirComp}
In the above example, we assumed that the upload of local model parameters to the \ac{ps} and download of global model from the \ac{ps} are performed noiselessly through orthogonal channels. This is however not the case in a wireless network. In wireless networks, multiple devices usually communicate with the \ac{ps} through a \ac{mac}. In the first glance, this seems to increase the computational load of the \ac{ps}: The \ac{ps} needs to first estimate individual local parameters and then combine them. Nevertheless, it is straightforward to see that this is a wrong conclusion. To clarify this point, consider the following simple setting: All devices transmit their local parameters over a static linear \ac{mac} in $D$ consequent transmission time slots. Assume that each device has a single antenna, and the \ac{ps} is equipped with an antenna array of size $N$.  Hence, after $D$ time slots, the \ac{ps} receives
\begin{align}
	\mY = \mH \mTheta^\trp + \mXi,
\end{align}
where $\mXi$ denotes the noise process, $\mH \in \setC^{N\times K}$ is the channel matrix whose entry $\brc{n,k}$ represents the channel coefficient between device $k$ and antenna $n$,  
\begin{align}
	\mTheta = \dbc{\btheta_1,\ldots,\btheta_K},
\end{align}
and $\mY \in \setC^{N\times D}$ is the matrix of received signals whose column $d$ represents the signal received by the \ac{ps} array antenna in the $d$-th channel use. 

From the channel model, it is observed that the received signals are already combined versions of the local parameters. Hence, a proper linear receiver can estimate $\bar{\btheta}$ directly from $\mY$ without recovering the individual local parameters. To further clarify this latter point, let us ignore the noise process in the channel. Using a linear receiver $\br\in\setC^{N}$, the signals received over multiple antennas of the \ac{ps} can be combined into a single vector $\hat{\btheta}\in\setC^D$ as
\begin{align}
	\hat{\btheta}^\trp = \br^\trp \mY = \br^\trp \mH \mTheta^\trp.
\end{align}
By designing $\br$ such that 
\begin{align}
	\br^\trp \mH = \dbc{\phi_1,\ldots,\phi_K}, \label{eq:OAF_1}
\end{align}
we have
\begin{align}
	\hat{\btheta} =   \mTheta \dbc{\phi_1,\ldots,\phi_K}^\trp = \bar{\btheta}.
\end{align}
This means that the \ac{ps} can directly apply federated averaging by finding a $\br$ which satisfies \eqref{eq:OAF_1} and no further parameter estimation is required. With a large-enough antenna array, such a task is easily fulfilled via \ac{zf}.

In practice, the exact calculation of $\bar{\btheta}$ is not possible at the \ac{ps}, as the received signals are noisy. Thus, in this case, we should determine the linear receiver $\br$ such that $\hat{\btheta}$ estimates $\bar{\btheta}$, properly. A common approach is to find the receiver by applying an \ac{mmse} estimator to determine the target model parameters $\bar{\btheta}$. This means that we find $\br$ as 
\begin{align}
	\br =   \argmin_{\bxx \in \setR^N} \Ex{\norm{\mY^\trp \bxx - \bar{\btheta}}^2 }{ }.
\end{align}
Here, the expectation is taken with respect to all random quantities in the channel, as well as the distribution of $\bar{\btheta}$ which is determined through assuming a stochastic model for either the local datasets or the local model parameters.

\section{MMSE vs. Zero-Forcing Coordination}
\label{app:MMSE}
As indicated in Section~\ref{sec:ZF}, zero forcing is a \textit{sub-optimal}, yet efficient, approach for coordination. The optimal approach in this case is to find \ac{mmse} estimate: We find the minimizer of the following objective function:
\begin{align}
	\mae \brc{\bmm , \setS, \eta , \bpsi} = \Ex{\abs{\hat{\theta} \brc{\bmm , \setS, \eta , \bpsi} -{\theta}\brc{\setS} }^2 }{ }.
\end{align}
with respect to all design parameters, i.e., $\bmm$, $\setS$, $\eta$, and $\bpsi$, subject to the power constraint, i.e., 
\begin{align}
	\norm{\bpsi}_{\infty} \leq \sqrt{P}
\end{align}
where $\norm{\cdot}_{\infty}$ denotes the $\ell_{\infty}$-norm. 

Unlike zero forcing coordination, the \ac{mmse} coordination does not necessarily cancel the local parameters in the objective function. Hence, the exact derivation of the metric in this case requires an explicit statistical model for the local model parameters. This leads to a more complicated expression for the aggregation cost metric. It is worth mentioning that, even for the \textit{unrealistic} model of \ac{iid} model parameters, a closed-form expression for the objective function is not straightforwardly derived for an arbitrary noise variance. It is however straightforward to show that for small enough noise variances, the \ac{mmse} and zero forcing coordination approaches meet.

\section{Extensions to Parametrized Channels}
\label{app:A}
In various recent technologies, the wireless channel is parametrized by tunable variables. One of the most well-known examples is \ac{irs}-aided networks in which the end-to-end channels between the \ac{ps} and the edge devices are modified by \acp{irs} \cite{Wang2022FLIRS}. The efficient approach in these networks is to tune the channel parameters jointly with user scheduling. For the particular example of \ac{irs}-aided networks, a \ac{dc} programming-based algorithm for joint scheduling and channel tuning is proposed in \cite{Wang2022FLIRS}. Considering the wide scope of recent technologies that enable the parametrization of wireless channels, in this section, we extend our proposed scheduling scheme to a more generic form by which joint scheduling and channel tuning is performed.

\subsection{Networks with Parametrized Channels}
A parametrized channel between edge device $k$ and the \ac{ps} is in general denoted by $\bh_k\brc{\bmu}$ for some vector of tunable parameters
\begin{align}
	\bmu = \dbc{\mu_1 , \ldots, \mu_M}^\trp \label{eq:mu}
\end{align}
with support $\setU$, i.e., $\bmu\in\setU$. These parametrized channels can be collected into a parametrized channel matrix as
\begin{align}
	\mH \brc{\bmu} =  \dbc{\bh_1\brc{\bmu},\ldots,\bh_K\brc{\bmu}}.
\end{align}

A well-known example of parametrized channels is a wireless network equipped with \acp{irs}. This example is illustrated below:
\begin{example}[IRS-aided Networks]
	\label{ex:IRS}
	Consider the same setting presented throughout the problem formulation in Section~\ref{sec:system-model}. We further assume that a passive \ac{irs} is located at a fixed distance from the \ac{ps}. The \ac{irs} consists of $M$ reflecting elements that reflect their received signals after applying tunable phase-shifts. Let $\rmg_{m,k}$ denote the channel coefficient between device $k$ and the $m$-th reflecting element on the \ac{irs}. Moreover, denote the channel vector between \ac{irs} element $m$ and the \ac{ps} by $\bt_m\in\setC^N$. The end-to-end channel between device $k$ and the \ac{ps} is given by
	\begin{align}
		\bh_k\brc{\bmu} = \bh_k^0 + \mT \mG_k \bmu
	\end{align}
	where $\bh_k^0 $ denotes the direct link between the \ac{ps} and device $k$, the matrix $\mT$ represents the channel matrix between the \ac{irs} and the \ac{ps}, i.e., $\mT = \dbc{\bt_1,\ldots,\bt_M}$, and 
	\begin{align}
		\mG_k = \Diag{\rmg_{1,k} , \ldots, \rmg_{M,k}}.
	\end{align}
	The parameter vector $\bmu$ is defined as in \eqref{eq:mu} with $\mu_m$ denoting the \textit{tunable} phase-shift applied by \ac{irs} element $m$, i.e., $\mu_m \in \setU_1$  with $\setU_1$ denoting the unit circle on the complex plane.
\end{example}

\subsection{Joint Scheduling and Tuning via Matching Pursuit}
In a network with parametrized channels, the computation error directly depends on the channel parameters: For a given $\bmu$, the \textit{parametrized} computation error is given by
\begin{align}
	\mae \brc{\bmu, \bmm, \setS}=  \frac{\sigma^2}{P}\max_{k\in \setS} \phi_k^2 \dfrac{\norm{\bmm}^2}{\abs{\bh_k^\her\brc{\bmu} \bmm}^2}. 
\end{align}
As a result, the scheduling and channel tuning are mutually coupled: For a given subset of devices, the channel tuning task aims to find $\bmu$, such that the parametrized computation error is minimized. For a given set of channel parameters, scheduling tries to find the largest subset of devices that results in a bounded computation error. 

The joint scheduling and channel tuning task can be mathematically formulated as
\begin{align}
	&\max_{\bmu\in\setU^M , \bmm\in\setC^N, \setS} \; \abs{\setS}  \label{eq:JointSchTune} \\
	&\begin{array}{lll}
		\subto & \phi_k^2 {\norm{\bmm}^2} - \gamma {\abs{\bh_k^\her\brc{\bmu} \bmm}^2} \leq 0  & \forall k\in\setS \\
		& \setS \subseteq \dbc{K} &
	\end{array}\nonumber	.
\end{align}
Similar to the basic case with non-parametrized channels, the direct joint solution to \eqref{eq:JointSchTune} is computationally intractable, as the scheduling task reduces to integer programming. We hence extend the matching pursuit approach to this case while tuning the channel parameters via the alternating optimization technique.

We start the extension, by focusing on iteration $t$ in which the subset $\setS\itr{t}$ is determined after omitting index $i\itr{t}$ from $\setS\itr{t-1}$ in the previous iteration. Denote the channel parameters updated in the last iteration by $\bmu\itr{t-1}$. Following the similar weighting approach as in Algorithm~\ref{alg:MPS}, the linear operator $\bmm$ and the channel parameters $\bmu$ are jointly updated as
\begin{align}
	\brc{\bmu\itr{t} ,  \bmm\itr{t}} &= \argmin_{\bmu, \bmm} \; \bmm^\her \brc{\Phi\itr{t} \mI_N  - \gamma \mH\brc{\bmu} \mW\itr{t} \mH^\her\brc{\bmu} } \bmm. \label{eq:mu_c}
\end{align}
The global solution to this optimization problem is in general intractable. We hence invoke the alternating optimization technique to approximate the solution.

The alternating optimization approach approximates the solution as follows: It starts by setting $\bmu_0 = \bmu\itr{t-1}$ and finds $\bmm_j$ for alternation $j\geq 0$ as 
\begin{subequations}
	\begin{align}
		\bmm_{j+1} &= \argmin_{ \bmm} \; \bmm^\her \brc{\Phi\itr{t} \mI_N  - \gamma \mH\brc{\bmu_{j}} \mW\itr{t} \mH^\her\brc{\bmu_{j}} } \bmm\\
		&=  \bvv_{j} \label{eq:c_update}
	\end{align}
\end{subequations}
where $\bvv_{j}$ is the column corresponding to the largest singular value of $\mH\brc{\bmu_{j}}\sqrt{\mW\itr{t}}$, i.e., the $n^\star$-th column of $\mV_{j}$ with
\begin{align}
	\mH\brc{\bmu_{j}} \sqrt{\mW\itr{t}} = \mV_{j} \mSigma_{j} \mU_{j}^{ \her}
\end{align}
and $n^\star$ denotes the index of the column of $\mSigma_{j}$ which includes the largest non-zero entry. It then alternates the optimization variable, i.e., it solves 
\begin{subequations}
	\begin{align}
		\bmu_{j+1} &= \argmin_{ \bmu \in \setU^M} \; \bmm_{j+1}^\her \brc{\Phi\itr{t} \mI_N  - \gamma \mH\brc{\bmu} \mW\itr{t} \mH^\her\brc{\bmu} } \bmm_{j+1},\\
		&= \argmax_{ \bmu \in \setU^M} \; \bmm_{j+1}^\her \mH\brc{\bmu} \mW\itr{t} \mH^\her\brc{\bmu} \bmm_{j+1},
	\end{align}
\end{subequations}
which is a norm minimization and can be solved by various algorithms depending on the exact form of $\mH\brc{\bmu}$. For many forms of $\mH\brc{\bmu}$, the solution to this optimization cannot be determined explicitly ans is approximated via some algorithm. We hence denote it as 	
\begin{align}
	\bmu_{j+1} = \maP \brc{\bmm_{j+1}\vert \mH, \mW\itr{t}} \label{eq:mu_update}
\end{align}
where $\maP\brc{\cdot}$ denotes the algorithm which finds (or approximates) 
\begin{align}
	\bmu_{j+1}^\star = \argmax_{ \bmu \in \setU^M} \; \bmm_{j+1}^\her \mH\brc{\bmu} \mW\itr{t} \mH^\her\brc{\bmu} \bmm_{j+1}.
\end{align}
We call this algorithm the \textit{channel update policy} which should be adapted to the application.

The solution to \eqref{eq:mu_c} is finally approximated by alternating between \eqref{eq:mu_update} and \eqref{eq:c_update} for a finite number of alternations. The receiver and channel parameters in iteration $t$ of the main loop, i.e., $\bmm\itr{t}$ and $\bmu\itr{t}$, are set to be the converging solutions determined by alternating optimization.

We finally find the most dominant root of the computation error and remove it from the selected support, i.e., we set
\begin{align}
	i\itr{t+1} = \argmax_{k \in \setS\itr{t} }  \phi_k^2  - \gamma {\abs{\bh_k^\her \brc{\bmu\itr{t}}  \bmm\itr{t} }^2},
\end{align}
and terminate the algorithm, when the computation error falls below the desired threshold. The pseudo-code for this process is given in Algorithm~\ref{alg:MPS_Tune}.

\begin{algorithm}[H]
	\caption{Matching Pursuit Based Joint Tuning and Scheduling $\maA_2\brc{\cdot}$} 
	\label{alg:MPS_Tune}
	\begin{algorithmic}
		\INPUT The updating strategy $\Pi\brc{\cdot}$, and the real positive scalars $\phi_k$ for $k\in \dbc{K}$.
		\REQUIRE{ Set $\setS\itr{0} = \dbc{K}$, $i\itr{1} = \emptyset$, $\mW\itr{1} = \mI_K$, $\Delta\itr{1} = +\infty$ and $F_k\itr{1} = 1$ for $k\in\dbc{K}$. Set $\bmu\itr{0}$ to some feasible point in $\setU^M$.}
		\WHILE{ $\Delta\itr{t} > 0$ }{
			\STATE Set $\setS\itr{t} = \setS\itr{t-1} - \set{i\itr{t}}$.
			\STATE Update $\mW\itr{t}$ by setting $\dbc{\mW\itr{t}}_{i\itr{t},i\itr{t}} = 0$ and
				$\dbc{\mW\itr{t}}_{k,k} = \Pi \brc{ F_k\itr{t} }$
			for all $k\in\setS\itr{t}$.
			\STATE Set $\bmu_0 = \bmu\itr{t-1}$ and $j=0$.
			\WHILE{ It converges}{
				\STATE Find the \ac{svd} %
					$\mH \brc{\bmu_{j}} \sqrt{\mW\itr{t}} = \mV_{j} \mSigma_j \mU_j^{ \her}$ %
				and let $n_j$ be the index of the largest singular value.
				\STATE Let $\bmm_{j+1} $ be the $n_j$-th column of $ \mV_j$ and update the channel parameters as
				\begin{align}
					\bmu_{j+1} = \maP \brc{\bmm_{j+1}\vert \mH, \mW\itr{t}}
				\end{align}
				\STATE Let $j \leftarrow j+1$.
			}
			\ENDWHILE
			\STATE Set $\bmu\itr{t} = \bmu_{J}$ and $\bmm\itr{t} = \bmm_{J}$ with $J$ denoting the final alternation.
			\STATE Update the \textit{constraint indicator} $F_k\itr{t}$ for device $k \in \setS\itr{t} $ as
			\begin{align}
				F_k\itr{t}  = \phi_k^2  - \gamma {\abs{\bh_k^\her\brc{\bmu\itr{t}} \bmm\itr{t} }^2}.
			\end{align}
			\STATE Find the next device index and the maximal constraint as
			\begin{subequations}
				\begin{align}
					i\itr{t+1} &= \argmax_{k \in \setS\itr{t} } F_k\itr{t},\\
					\Delta\itr{t+1} &=  \max_{k \in \setS\itr{t} } F_k\itr{t}.
				\end{align}
			\end{subequations}
			\STATE Let $t \leftarrow t+1$.
		}
		\ENDWHILE
		\OUTPUT $\bmm\itr{T}$, $\bmu\itr{T}$ and $\setS\itr{T}$ with $T$ being the last iteration.
	\end{algorithmic} 
\end{algorithm}

\subsection{Example: OTA-FL via AirComp in IRS-Aided Wireless Networks}
The joint scheduling and tuning algorithm considers a general setting with parametrized channels. Nevertheless, for a given application, the channel update policy $\maP\brc{\cdot}$ should be adapted due to the parametrization model $\mH\brc{\bmu}$. In this section, we consider the particular example of \ac{irs}-aided wireless networks and develop a channel update policy for the joint scheduling and channel tuning algorithm.

Consider the \ac{irs}-aided network described in Example~\ref{ex:IRS}. For this setting, the marginal channel tuning problem reduces to
\begin{align}
	\bmu_{j+1}^\star = \argmax_{ \bmu \in \setU_1^M} \; \bmu^\her \mQ_j\itr{t} \bmu + 2 \Re\set{\bmu^\her \baa_j\itr{t}}
\end{align} 
where $\mQ_j\itr{t}$ and $\baa_j\itr{t}$ are given by
\begin{subequations}
	\begin{align}
		\mQ_j\itr{t} &= \sum_{k=1}^K w_k\itr{t} \mG_k^\her \mT \bmm_{j+1} \bmm_{j+1}^\her \mT \mG_k,\\
		\baa_j\itr{t} &= \sum_{k=1}^K w_k\itr{t} \mG_k^\her \mT \bmm_{j+1} \bmm_{j+1}^\her \bh^0_k.
	\end{align}
\end{subequations}
This is a unit-modulus problem which reduces in general to an \ac{np}-hard problem. Classical approaches to estimate the solutions in polynomial time are to use the \ac{bcd} method or \ac{mm} technique.

\bibliography{ref}

\begin{thebibliography}{10}
\providecommand{\url}[1]{#1}
\csname url@samestyle\endcsname
\providecommand{\newblock}{\relax}
\providecommand{\bibinfo}[2]{#2}
\providecommand{\BIBentrySTDinterwordspacing}{\spaceskip=0pt\relax}
\providecommand{\BIBentryALTinterwordstretchfactor}{4}
\providecommand{\BIBentryALTinterwordspacing}{\spaceskip=\fontdimen2\font plus
\BIBentryALTinterwordstretchfactor\fontdimen3\font minus
  \fontdimen4\font\relax}
\providecommand{\BIBforeignlanguage}[2]{{%
\expandafter\ifx\csname l@#1\endcsname\relax
\typeout{** WARNING: IEEEtran.bst: No hyphenation pattern has been}%
\typeout{** loaded for the language `#1'. Using the pattern for}%
\typeout{** the default language instead.}%
\else
\language=\csname l@#1\endcsname
\fi
#2}}
\providecommand{\BIBdecl}{\relax}
\BIBdecl

\bibitem{lyu2020threats}
L.~Lyu, H.~Yu, and Q.~Yang, ``Threats to federated learning: A survey,'' 2020,
  arXiv preprint arXiv:2003.02133.

\bibitem{aledhari2020federated}
M.~Aledhari, R.~Razzak, R.~M. Parizi, and F.~Saeed, ``Federated learning: A
  survey on enabling technologies, protocols, and applications,'' \emph{IEEE
  Access}, vol.~8, pp. 140\,699--140\,725, 2020.

\bibitem{liang2020think}
P.~P. Liang, T.~Liu, L.~Ziyin, N.~B. Allen, R.~P. Auerbach, D.~Brent,
  R.~Salakhutdinov, and L.-P. Morency, ``Think locally, act globally: Federated
  learning with local and global representations,'' 2020, arXiv preprint
  arXiv:2001.01523.

\bibitem{yang2021federated}
Z.~Yang, M.~Chen, K.-K. Wong, H.~V. Poor, and S.~Cui, ``Federated learning for
  {6G}: Applications, challenges, and opportunities,'' \emph{Engineering},
  2021.

\bibitem{nguyen2021federated}
D.~C. Nguyen, M.~Ding, P.~N. Pathirana, A.~Seneviratne, J.~Li, and H.~V. Poor,
  ``Federated learning for internet of things: A comprehensive survey,''
  \emph{IEEE Communications Surveys \& Tutorials}, 2021.

\bibitem{ananthanarayanan2017real}
G.~Ananthanarayanan, P.~Bahl, P.~Bod{\'\i}k, K.~Chintalapudi, M.~Philipose,
  L.~Ravindranath, and S.~Sinha, ``Real-time video analytics: The killer app
  for edge computing,'' \emph{Computer}, vol.~50, no.~10, pp. 58--67, Oct.
  2017.

\bibitem{hung2018videoedge}
C.-C. Hung, G.~Ananthanarayanan, P.~Bodik, L.~Golubchik, M.~Yu, P.~Bahl, and
  M.~Philipose, ``Videoedge: Processing camera streams using hierarchical
  clusters,'' in \emph{Proc. IEEE/ACM Symposium on Edge Computing (SEC)},
  Seattle, WA,USA, Oct. 2018, pp. 115--131.

\bibitem{zhou2019edge}
Z.~Zhou, X.~Chen, E.~Li, L.~Zeng, K.~Luo, and J.~Zhang, ``Edge intelligence:
  Paving the last mile of artificial intelligence with edge computing,''
  \emph{Proceedings of the IEEE}, vol. 107, no.~8, pp. 1738--1762, 2019.

\bibitem{shi2016edge}
W.~Shi, J.~Cao, Q.~Zhang, Y.~Li, and L.~Xu, ``Edge computing: Vision and
  challenges,'' \emph{IEEE Internet of Things Journal}, vol.~3, no.~5, pp.
  637--646, 2016.

\bibitem{shi2016promise}
W.~Shi and S.~Dustdar, ``The promise of edge computing,'' \emph{Computer},
  vol.~49, no.~5, pp. 78--81, 2016.

\bibitem{abbas2017mobile}
N.~Abbas, Y.~Zhang, A.~Taherkordi, and T.~Skeie, ``Mobile edge computing: A
  survey,'' \emph{IEEE Internet of Things Journal}, vol.~5, no.~1, pp.
  450--465, 2017.

\bibitem{jordan2018communication}
M.~I. Jordan, J.~D. Lee, and Y.~Yang, ``Communication-efficient distributed
  statistical inference,'' \emph{Journal of the American Statistical
  Association}, vol. 114, no. 526, pp. 668--681, 2019.

\bibitem{chen2018lag}
T.~Chen, G.~Giannakis, T.~Sun, and W.~Yin, ``{LAG}: Lazily aggregated gradient
  for communication-efficient distributed learning,'' \emph{Advances in Neural
  Information Processing Systems}, vol.~31, 2018.

\bibitem{elgabli2020gadmm}
A.~Elgabli, J.~Park, A.~S. Bedi, M.~Bennis, and V.~Aggarwal, ``{GADMM}: Fast
  and communication efficient framework for distributed machine learning.''
  \emph{Journal of Machine Learning Research}, vol.~21, no.~76, pp. 1--39,
  2020.

\bibitem{liu2020over}
W.~Liu, X.~Zang, Y.~Li, and B.~Vucetic, ``Over-the-air computation systems:
  Optimization, analysis and scaling laws,'' \emph{IEEE Transactions on
  Wireless Communications}, vol.~19, no.~8, pp. 5488--5502, 2020.

\bibitem{li2020federated}
T.~Li, A.~K. Sahu, A.~Talwalkar, and V.~Smith, ``Federated learning:
  Challenges, methods, and future directions,'' \emph{IEEE Signal Processing
  Magazine}, vol.~37, no.~3, pp. 50--60, 2020.

\bibitem{konevcny2016federated}
J.~Kone{\v{c}}n{\`y}, H.~B. McMahan, D.~Ramage, and P.~Richt{\'a}rik,
  ``Federated optimization: Distributed machine learning for on-device
  intelligence,'' 2016, arXiv preprint arXiv:1610.02527.

\bibitem{bonawitz2019towards}
K.~Bonawitz, H.~Eichner, W.~Grieskamp, D.~Huba, A.~Ingerman, V.~Ivanov,
  C.~Kiddon, J.~Kone{\v{c}}n{\`y}, S.~Mazzocchi, B.~McMahan \emph{et~al.},
  ``Towards federated learning at scale: System design,'' \emph{Proceedings of
  Machine Learning and Systems}, vol.~1, pp. 374--388, 2019.

\bibitem{yang2019federated}
Q.~Yang, Y.~Liu, Y.~Cheng, Y.~Kang, T.~Chen, and H.~Yu, ``Federated learning,''
  \emph{Synthesis Lectures on Artificial Intelligence and Machine Learning},
  vol.~13, no.~3, pp. 1--207, 2019.

\bibitem{sattler2019robust}
F.~Sattler, S.~Wiedemann, K.-R. M{\"u}ller, and W.~Samek, ``Robust and
  communication-efficient federated learning from non-iid data,'' \emph{IEEE
  Transactions on Neural Networks and Learning Systems}, vol.~31, no.~9, pp.
  3400--3413, 2019.

\bibitem{shoham2019overcoming}
N.~Shoham, T.~Avidor, A.~Keren, N.~Israel, D.~Benditkis, L.~Mor-Yosef, and
  I.~Zeitak, ``Overcoming forgetting in federated learning on non-iid data,''
  2019, arXiv preprint arXiv:1910.07796.

\bibitem{lee2020bayesian}
S.~Lee, C.~Park, S.-N. Hong, Y.~C. Eldar, and N.~Lee, ``Bayesian federated
  learning over wireless networks,'' 2020, arXiv preprint arXiv:2012.15486.

\bibitem{luo2021no}
M.~Luo, F.~Chen, D.~Hu, Y.~Zhang, J.~Liang, and J.~Feng, ``No fear of
  heterogeneity: Classifier calibration for federated learning with non-iid
  data,'' \emph{Advances in Neural Information Processing Systems}, vol.~34,
  2021.

\bibitem{gafni2022federated}
T.~Gafni, N.~Shlezinger, K.~Cohen, Y.~C. Eldar, and H.~V. Poor, ``Federated
  learning: A signal processing perspective,'' \emph{IEEE Signal Processing
  Magazine}, vol.~39, no.~3, pp. 14--41, 2022.

\bibitem{fredrikson2015model}
M.~Fredrikson, S.~Jha, and T.~Ristenpart, ``Model inversion attacks that
  exploit confidence information and basic countermeasures,'' in \emph{Proc. of
  the 22nd ACM SIGSAC Conference on Computer and Communications Security},
  2015, pp. 1322--1333.

\bibitem{wei2020secFed}
K.~Wei, J.~Li, M.~Ding, C.~Ma, H.~H. Yang, F.~Farokhi, S.~Jin, T.~Q.~S. Quek,
  and H.~V. Poor, ``Federated learning with differential privacy: Algorithms
  and performance analysis,'' \emph{IEEE Transactions on Information Forensics
  and Security}, vol.~15, pp. 3454--3469, 2020.

\bibitem{kim2021federated}
M.~Kim, O.~G{\"u}nl{\"u}, and R.~F. Schaefer, ``Federated learning with local
  differential privacy: Trade-offs between privacy, utility, and
  communication,'' in \emph{Proc. IEEE International Conference on Acoustics,
  Speech and Signal Processing (ICASSP)}, 2021, pp. 2650--2654.

\bibitem{shi2020communication}
Y.~Shi, K.~Yang, T.~Jiang, J.~Zhang, and K.~B. Letaief,
  ``Communication-efficient edge {AI}: Algorithms and systems,'' \emph{IEEE
  Communications Surveys \& Tutorials}, vol.~22, no.~4, pp. 2167--2191, 2020.

\bibitem{yang2020scheduling}
H.~H. Yang, Z.~Liu, T.~Q.~S. Quek, and H.~V. Poor, ``Scheduling policies for
  federated learning in wireless networks,'' \emph{IEEE Transactions on
  Communications}, vol.~68, no.~1, pp. 317--333, 2020.

\bibitem{chen2021communication}
M.~Chen, N.~Shlezinger, H.~V. Poor, Y.~C. Eldar, and S.~Cui,
  ``Communication-efficient federated learning,'' \emph{Proceedings of the
  National Academy of Sciences}, vol. 118, no.~17, 2021.

\bibitem{crane2019dingo}
R.~Crane and F.~Roosta, ``{DINGO}: Distributed {Newton}-type method for
  gradient-norm optimization,'' \emph{Advances in Neural Information Processing
  Systems}, vol.~32, 2019.

\bibitem{gao2020can}
H.~Gao and H.~Huang, ``Can stochastic zeroth-order frank-wolfe method converge
  faster for non-convex problems?'' in \emph{Proc. International Conference on
  Machine Learning}.\hskip 1em plus 0.5em minus 0.4em\relax PMLR, 2020, pp.
  3377--3386.

\bibitem{nguyen2020fast}
H.~T. Nguyen, V.~Sehwag, S.~Hosseinalipour, C.~G. Brinton, M.~Chiang, and H.~V.
  Poor, ``Fast-convergent federated learning,'' \emph{IEEE Journal on Selected
  Areas in Communications}, vol.~39, no.~1, pp. 201--218, 2020.

\bibitem{lin2018deep}
Y.~Lin, S.~Han, H.~Mao, Y.~Wang, and B.~Dally, ``Deep gradient compression:
  Reducing the communication bandwidth for distributed training,'' in
  \emph{Proc. International Conference on Learning Representations}, Feb. 2018.

\bibitem{yuan2018variance}
K.~Yuan, B.~Ying, J.~Liu, and A.~H. Sayed, ``Variance-reduced stochastic
  learning by networked agents under random reshuffling,'' \emph{IEEE
  Transactions on Signal Processing}, vol.~67, no.~2, pp. 351--366, Sep. 2018.

\bibitem{shlezinger2020uveqfed}
N.~Shlezinger, M.~Chen, Y.~C. Eldar, H.~V. Poor, and S.~Cui, ``{UVeQFed}:
  Universal vector quantization for federated learning,'' \emph{IEEE
  Transactions on Signal Processing}, vol.~69, pp. 500--514, 2020.

\bibitem{zhu2019broadband}
G.~Zhu, Y.~Wang, and K.~Huang, ``Broadband analog aggregation for low-latency
  federated edge learning,'' \emph{IEEE Transactions on Wireless
  Communications}, vol.~19, no.~1, pp. 491--506, Oct. 2019.

\bibitem{yang2020federated}
K.~Yang, T.~Jiang, Y.~Shi, and Z.~Ding, ``Federated learning via over-the-air
  computation,'' \emph{IEEE Transactions on Wireless Communications}, vol.~19,
  no.~3, pp. 2022--2035, Jan. 2020.

\bibitem{amiri2020machine}
M.~M. Amiri and D.~G{\"u}nd{\"u}z, ``Machine learning at the wireless edge:
  Distributed stochastic gradient descent over-the-air,'' \emph{IEEE
  Transactions on Signal Processing}, vol.~68, pp. 2155--2169, 2020.

\bibitem{xu2021learning}
C.~Xu, S.~Liu, Z.~Yang, Y.~Huang, and K.-K. Wong, ``Learning rate optimization
  for federated learning exploiting over-the-air computation,'' \emph{IEEE
  Journal on Selected Areas in Communications}, vol.~39, no.~12, pp.
  3742--3756, 2021.

\bibitem{nazer2007computation}
B.~Nazer and M.~Gastpar, ``Computation over multiple-access channels,''
  \emph{IEEE Transactions on Information Theory}, vol.~53, no.~10, pp.
  3498--3516, 2007.

\bibitem{shi2021joint}
W.~Shi, S.~Zhou, Z.~Niu, M.~Jiang, and L.~Geng, ``Joint device scheduling and
  resource allocation for latency constrained wireless federated learning,''
  \emph{IEEE Transactions on Wireless Communications}, vol.~20, no.~1, pp.
  453--467, 2021.

\bibitem{xia2021federated}
W.~Xia, W.~Wen, K.-K. Wong, T.~Q. Quek, J.~Zhang, and H.~Zhu,
  ``Federated-learning-based client scheduling for low-latency wireless
  communications,'' \emph{IEEE Wireless Communications}, vol.~28, no.~2, pp.
  32--38, 2021.

\bibitem{amiri2020update}
M.~M. Amiri, D.~G{\"u}nd{\"u}z, S.~R. Kulkarni, and H.~V. Poor, ``Update aware
  device scheduling for federated learning at the wireless edge,'' in
  \emph{Proc. IEEE International Symposium on Information Theory (ISIT)}.\hskip
  1em plus 0.5em minus 0.4em\relax IEEE, 2020, pp. 2598--2603.

\bibitem{chen2020wireless}
M.~Chen, H.~V. Poor, W.~Saad, and S.~Cui, ``Wireless communications for
  collaborative federated learning,'' \emph{IEEE Communications Magazine},
  vol.~58, no.~12, pp. 48--54, 2020.

\bibitem{yang2020energy}
Z.~Yang, M.~Chen, W.~Saad, C.~S. Hong, and M.~Shikh-Bahaei, ``Energy efficient
  federated learning over wireless communication networks,'' \emph{IEEE
  Transactions on Wireless Communications}, vol.~20, no.~3, pp. 1935--1949,
  2020.

\bibitem{Wang2022FLIRS}
Z.~Wang, J.~Qiu, Y.~Zhou, Y.~Shi, L.~Fu, W.~Chen, and K.~B. Letaief,
  ``Federated learning via intelligent reflecting surface,'' \emph{IEEE
  Transactions on Wireless Communications}, vol.~21, no.~2, pp. 808--822, Feb.
  2022.

\bibitem{mallat1993matching}
S.~G. Mallat and Z.~Zhang, ``Matching pursuits with time-frequency
  dictionaries,'' \emph{IEEE Transactions on Signal Processing}, vol.~41,
  no.~12, pp. 3397--3415, 1993.

\bibitem{liu2021reconfigurable}
H.~Liu, X.~Yuan, and Y.-J.~A. Zhang, ``Reconfigurable intelligent surface
  enabled federated learning: A unified communication-learning design
  approach,'' \emph{IEEE Transactions on Wireless Communications}, vol.~20,
  no.~11, pp. 7595--7609, 2021.

\bibitem{wadu2020federated}
M.~M. Wadu, S.~Samarakoon, and M.~Bennis, ``Federated learning under channel
  uncertainty: Joint client scheduling and resource allocation,'' in
  \emph{Proc. IEEE Wireless Communications and Networking Conference (WCNC)},
  2020, pp. 1--6.

\bibitem{goldenbaum2013robust}
M.~Goldenbaum and S.~Stanczak, ``Robust analog function computation via
  wireless multiple-access channels,'' \emph{IEEE Transactions on
  Communications}, vol.~61, no.~9, pp. 3863--3877, 2013.

\bibitem{goldenbaum2013harnessing}
M.~Goldenbaum, H.~Boche, and S.~Sta{\'n}czak, ``Harnessing interference for
  analog function computation in wireless sensor networks,'' \emph{IEEE
  Transactions on Signal Processing}, vol.~61, no.~20, pp. 4893--4906, 2013.

\bibitem{chen2018uniform}
L.~Chen, X.~Qin, and G.~Wei, ``A uniform-forcing transceiver design for
  over-the-air function computation,'' \emph{IEEE Wireless Communications
  Letters}, vol.~7, no.~6, pp. 942--945, 2018.

\bibitem{zhu2018mimo}
G.~Zhu and K.~Huang, ``{MIMO} over-the-air computation for high-mobility
  multimodal sensing,'' \emph{IEEE Internet of Things Journal}, vol.~6, no.~4,
  pp. 6089--6103, 2018.

\bibitem{zhai2021hybrid}
X.~Zhai, X.~Chen, J.~Xu, and D.~W.~K. Ng, ``Hybrid beamforming for massive
  {MIMO} over-the-air computation,'' \emph{IEEE Transactions on
  Communications}, vol.~69, no.~4, pp. 2737--2751, 2021.

\bibitem{cao2021optimized}
X.~Cao, G.~Zhu, J.~Xu, Z.~Wang, and S.~Cui, ``Optimized power control design
  for over-the-air federated edge learning,'' \emph{IEEE Journal on Selected
  Areas in Communications}, vol.~40, no.~1, pp. 342--358, 2021.

\bibitem{sery2020analog}
T.~Sery and K.~Cohen, ``On analog gradient descent learning over multiple
  access fading channels,'' \emph{IEEE Transactions on Signal Processing},
  vol.~68, pp. 2897--2911, 2020.

\bibitem{mcmahan2017communication}
B.~McMahan, E.~Moore, D.~Ramage, S.~Hampson, and B.~A. y~Arcas,
  ``Communication-efficient learning of deep networks from decentralized
  data,'' in \emph{Proc. PMLR Artificial Intelligence and Statistics}.\hskip
  1em plus 0.5em minus 0.4em\relax PMLR, 2017, USA, pp. 1273--1282.

\bibitem{tropp2017practical}
J.~A. Tropp, A.~Yurtsever, M.~Udell, and V.~Cevher, ``Practical sketching
  algorithms for low-rank matrix approximation,'' \emph{SIAM Journal on Matrix
  Analysis and Applications}, vol.~38, no.~4, pp. 1454--1485, 2017.

\bibitem{tao1997convex}
P.~D. Tao and L.~T.~H. An, ``Convex analysis approach to {DC} programming:
  theory, algorithms and applications,'' \emph{Acta Mathematica Vietnamica},
  vol.~22, no.~1, pp. 289--355, 1997.

\bibitem{le1997solving}
T.~H. Le~An and P.~Dinh~Tao, ``Solving a class of linearly constrained
  indefinite quadratic problems by {DC} algorithms,'' \emph{Journal of Global
  Optimization}, vol.~11, no.~3, pp. 253--285, 1997.

\bibitem{foucart2013invitation}
S.~Foucart and H.~Rauhut, \emph{A Mathematical Introduction to Compressive
  Sensing}.\hskip 1em plus 0.5em minus 0.4em\relax Springer, 2013.

\bibitem{donoho2006compressed}
D.~L. Donoho, ``Compressed sensing,'' \emph{IEEE Transactions on Information
  Theory}, vol.~52, no.~4, pp. 1289--1306, Apr. 2006.

\bibitem{candes2006robust}
E.~J. Cand{\`e}s, J.~Romberg, and T.~Tao, ``Robust uncertainty principles:
  Exact signal reconstruction from highly incomplete frequency information,''
  \emph{IEEE Transactions on Information Theory}, vol.~52, no.~2, pp. 489--509,
  Jan. 2006.

\bibitem{tropp2007signal}
J.~A. Tropp and A.~C. Gilbert, ``Signal recovery from random measurements via
  orthogonal matching pursuit,'' \emph{IEEE Transactions on Information
  Theory}, vol.~53, no.~12, pp. 4655--4666, Dec. 2007.

\bibitem{rebollo2002optimized}
L.~Rebollo-Neira and D.~Lowe, ``Optimized orthogonal matching pursuit
  approach,'' \emph{IEEE Signal Processing Letters}, vol.~9, no.~4, pp.
  137--140, Apr. 2002.

\bibitem{wang2012generalized}
J.~Wang, S.~Kwon, and B.~Shim, ``Generalized orthogonal matching pursuit,''
  \emph{IEEE Transactions on Signal Processing}, vol.~60, no.~12, pp.
  6202--6216, Sep. 2012.

\bibitem{cai2011orthogonal}
T.~T. Cai and L.~Wang, ``Orthogonal matching pursuit for sparse signal recovery
  with noise,'' \emph{IEEE Transactions on Information theory}, vol.~57, no.~7,
  pp. 4680--4688, Jun. 2011.

\bibitem{needell2009cosamp}
D.~Needell and J.~A. Tropp, ``{CoSaMP}: Iterative signal recovery from
  incomplete and inaccurate samples,'' \emph{Applied and Computational Harmonic
  Analysis}, vol.~26, no.~3, pp. 301--321, 2009.

\bibitem{parlett1982estimating}
B.~N. Parlett, H.~Simon, and L.~Stringer, ``On estimating the largest
  eigenvalue with the {Lanczos} algorithm,'' \emph{Mathematics of Computation},
  vol.~38, no. 157, pp. 153--165, 1982.

\bibitem{SCHWETLICK20031}
\BIBentryALTinterwordspacing
H.~Schwetlick and U.~Schnabel, ``Iterative computation of the smallest singular
  value and the corresponding singular vectors of a matrix,'' \emph{Linear
  Algebra and its Applications}, vol. 371, pp. 1--30, 2003. [Online].
  Available:
  \url{https://www.sciencedirect.com/science/article/pii/S0024379503004907}
\BIBentrySTDinterwordspacing

\bibitem{liang2014computing}
Q.~Liang and Q.~Ye, ``Computing singular values of large matrices with an
  inverse-free preconditioned {Krylov} subspace method,'' \emph{Electronic
  Transactions on Numerical Analysis}, vol.~42, p. 197, 2014.

\bibitem{horn2012matrix}
R.~A. Horn and C.~R. Johnson, \emph{Matrix Analysis}, 2nd~ed.\hskip 1em plus
  0.5em minus 0.4em\relax Cambridge University Press, 2013.

\bibitem{krizhevsky2009learning}
A.~Krizhevsky, G.~Hinton \emph{et~al.}, ``Learning multiple layers of features
  from tiny images,'' 2009, technical report, University of Toronto.

\bibitem{SIG-093}
E.~Björnson, J.~Hoydis, and L.~Sanguinetti, ``Massive {MIMO} networks:
  Spectral, energy, and hardware efficiency,'' \emph{Foundations and Trends®
  in Signal Processing}, vol.~11, no. 3-4, pp. 154--655, 2017.

\bibitem{li2022federated}
Q.~Li, Y.~Diao, Q.~Chen, and B.~He, ``Federated learning on non-iid data silos:
  An experimental study,'' in \emph{Proc. IEEE 38th International Conference on
  Data Engineering (ICDE)}, 2022, pp. 965--978.

\bibitem{simonyan2015very}
K.~Simonyan and A.~Zisserman, ``Very deep convolutional networks for
  large-scale image recognition,'' in \emph{Proc. International Conference on
  Learning Representations (ICLR)}, 2015, pp. 1--14, preprint at
  arXiv:1409.1556.

\bibitem{hwang2013holistic}
I.~Hwang, B.~Song, and S.~S. Soliman, ``A holistic view on hyper-dense
  heterogeneous and small cell networks,'' \emph{IEEE Communications Magazine},
  vol.~51, no.~6, pp. 20--27, 2013.

\bibitem{bhushan2014network}
N.~Bhushan, J.~Li, D.~Malladi, R.~Gilmore, D.~Brenner, A.~Damnjanovic, R.~T.
  Sukhavasi, C.~Patel, and S.~Geirhofer, ``Network densification: the dominant
  theme for wireless evolution into {5G},'' \emph{IEEE Communications
  Magazine}, vol.~52, no.~2, pp. 82--89, 2014.

\bibitem{andrews2016we}
J.~G. Andrews, X.~Zhang, G.~D. Durgin, and A.~K. Gupta, ``Are we approaching
  the fundamental limits of wireless network densification?'' \emph{IEEE
  Communications Magazine}, vol.~54, no.~10, pp. 184--190, 2016.

\bibitem{kamel2016ultra}
M.~Kamel, W.~Hamouda, and A.~Youssef, ``Ultra-dense networks: A survey,''
  \emph{IEEE Communications Surveys \& Tutorials}, vol.~18, no.~4, pp.
  2522--2545, May 2016.

\bibitem{adedoyin2020combination}
M.~A. Adedoyin and O.~E. Falowo, ``Combination of ultra-dense networks and
  other {5G} enabling technologies: A survey,'' \emph{IEEE Access}, vol.~8, pp.
  22\,893--22\,932, 2020.

\bibitem{kim2021energy}
S.~Kim, J.~Son, and B.~Shim, ``Energy-efficient ultra-dense network using
  {LSTM}-based deep neural networks,'' \emph{IEEE Transactions on Wireless
  Communications}, vol.~20, no.~7, pp. 4702--4715, 2021.

\bibitem{wu2019towards}
Q.~Wu and R.~Zhang, ``Towards smart and reconfigurable environment: Intelligent
  reflecting surface aided wireless network,'' \emph{IEEE Communications
  Magazine}, vol.~58, no.~1, pp. 106--112, 2019.

\end{thebibliography}
\bibliographystyle{IEEEtran}
\end{document}